\documentclass[12pt,draftcls,onecolumn,journal]{IEEEtran}
\makeatother
\usepackage{inputenc}
\usepackage{amssymb,amsthm}
\usepackage{amsfonts,amsmath}
\usepackage{latexsym}
\usepackage[mathscr]{euscript}
\usepackage{graphics}
\usepackage{graphicx,subfigure}
\usepackage{cite}
\usepackage{cases}
\usepackage{enumerate}
\usepackage{xcolor}

\usepackage{cancel}

\newcommand{\beq}{\begin{equation}}
\newcommand{\eeq}{\end{equation}}

\newcommand{\beqq}{\begin{equation*}}
\newcommand{\eeqq}{\end{equation*}}
\newcommand{\ei}{\end{itemize}}
\newcommand{\bi}{\begin{itemize}}
\newcommand{\ee}{\end{enumerate}}
\newcommand{\be}{\begin{enumerate}}

\newtheorem{prop}{Proposition}

\newtheorem{lemma}{Lemma}
\newtheorem{corol}{Corollary}
\newtheorem{algorithm}{Algorithm}

\theoremstyle{remark}

\newtheorem{remark}{Remark}

\newcommand{\argmax}[1]{\arg{\hbox{$\underset{#1}{\max}\,$}}}

\begin{document}

\sloppy

\title{FEAT: Fair Coordinated Iterative \\Water-Filling Algorithm}

\author{~Majed~Haddad\IEEEauthorrefmark{1},~Piotr~Wiecek\IEEEauthorrefmark{2},~Oussama~Habachi\IEEEauthorrefmark{3},~Samir~M.~Perlaza\IEEEauthorrefmark{4} and~Shahid M. Shah\IEEEauthorrefmark{5}\\
\IEEEauthorrefmark{1}University of Avignon, France\\
\IEEEauthorrefmark{2}Wroc{\l}aw University of Science and Technology, Wroc{\l}aw, Poland\\
\IEEEauthorrefmark{3}University of Clermont Auvergne, France\\
\IEEEauthorrefmark{4}INRIA, Sophia Antipolis, France\\
\IEEEauthorrefmark{5}National Institute of Technology, Srinagar, India\\
}

\maketitle

\begin{abstract}
In this paper, we consider a perfect coordinated water-filling game, where each user transmits solely on a given carrier. The main goal of the proposed algorithm (which we call FEAT) is to get close to the optimal, while keeping a decent level of fairness. The key idea within FEAT is to minimize the ratio between the best and the worst utilities of the users. This is done by ensuring that, at each iteration (channel assignment), a user is satisfied with this assignment as long as he does not loose much more than other users in the system. It has been shown that FEAT outperforms most related algorithms in many aspects, especially in interference-limited systems. Indeed, with FEAT we can ensure a near-optimal, fair and energy efficient solution with low computational complexity. In terms of robustness, it turns out that the balance between being nearly globally optimal and good from individual point of view seems hard to sustain with a significant number of users.
Also notice that, in this regard, global optimality gets less affected than the individual one, which offers hope that such an accurate water-filling algorithm can be designed around competition in interference-limited systems.

\end{abstract}

\vspace{.5cm}

\section{Introduction}\label{sec:intro}

In this paper, we tackle the well-known water-filling solution which solves the
problem of maximizing the information rate between the input and the output of a channel composed of several subchannels (such as a frequency-selective channel, a time-varying
channel) with an average power
constraint at the transmitters \cite{Gallager68,cover}. More specifically, we consider a multiple access channel (MAC), where $N$ transmitters wish to communicate with a common receiver over $K$ independent channels.

The purpose of this paper is twofold: On the one hand, it provides a nearly-optimal though low-complexity implementation in practice. To be more precise, it provides a practical algorithm with $O(NK^2\log(K))$ complexity. On the other hand, it guarantees fairness among the users by trying to assign to every user in the system at least one good channel.

\section*{Related Work}

Power control is well-known to be useful for the information rate enhancement. In fact, the main aim of power control is to select power levels corresponding to fading gain such that the sum information rate is optimized. Water-filling based power control algorithms are among the possible algorithms for power allocation. There have been many works on water-filling, but most of them do not consider fairness issue. They rather consider transmission rate-type utilities \cite{He13}. Water-filling based power control policies have been widely tied with the use of orthogonal frequency division multiplexing (OFDM) in wireless communications to mitigate the multiple path effect. As an example, water-filling has been used in the context of OFDM cognitive radio systems \cite{Qi12}, LTE systems \cite{Zhang14} and LTE-Advanced systems \cite{Zhang15}. Authors of \cite{Hasan20} proposed a Nash equilibrium based water-filling algorithm for a Multiple Input Multiple Output (MIMO) system in underground tunnel.
An iterative water-filling algorithm was proposed for non-orthogonal multiple access (NOMA) \cite{Cai17}. \cite{Youssef17} proposed a low complexity water-filling power allocation algorithm applied to the proportional fair scheduler for downlink NOMA systems. Authors of \cite{Gu17} proposed a transmission power allocation policy for HetNets using a water-filling algorithm.
In \cite{Fengfeng21}, authors proposed an energy-efficient power allocation for D2D communications underlaying cellular networks using a quasi-water-filling approach. 

In \cite{palomar05}, authors proposed an algorithm to evaluate a water-filling solution by considering a general water-filling formulation with multiple water-levels and multiple constraints. \cite{Jindal05} considered the problem of maximizing sum rate of a multiple-antenna Gaussian broadcast channel. The problem of optimizing power allocation in order to maximize mutual information has been solved in \cite{Lozano06} using the mercury water-filling principle. A generalization of
the mercury water-filling algorithm for multiple-input multiple output
(MIMO) Gaussian channels, while taking into account
the interference among inputs, was given by \cite{Perez10}. In \cite{Cao20}, authors proposed a water-filling based power allocation for Gaussian channels and approximately Gaussian inputs. In \cite{gaoning08}, authors extended the mercury water-filling policy to the multi-user context. The key idea behind mercury water-filling is to increase/decrease the water-levels to allow good/bad users to transmit more often than in the case of {\it{classical}} water-filling. Note that mercury water-filling enables us to design an efficient power allocation. However, the difficulty of obtaining closed-form analytical expressions, in the multiple users multiple channel model, often makes computing the mercury water-filling solution challenging. The algorithm introduced in this paper called FEAT is similar to mercury water-filling, but the way it is done is not by changing the water-filling level, but eliminating/adding available carriers on which users could have transmitted in the case without FEAT.

This paper, instead of dealing with {\it{classical}} multi-user information rate optimization problems resulting in significantly complicated water-filling solutions that include multiple water-levels and multiple constraints, has considered a different way to switch off transmission in users' links which do not contribute enough to the information rate to outweigh the interference degradation caused by them to the rest of the system. This is done by iteratively limiting the number of available carriers for each user instead of changing the water-levels as it is done in the above-mentioned water-filling problems.

The main goal of FEAT is to maximize the information rate, while keeping a decent level of fairness. More specifically, FEAT has the following benefits:
\bi
\item It minimizes the ratio between the best and the worst utilities of the users. We do that by ensuring that, at each iteration (channel assignment), a user is satisfied with his assignment as long as he does not loose much more than other users do,
\item It tries to give each user his best carrier (the one that gives him the highest utility) whenever it is possible, \emph{i.e.}, there are enough channels for every user and this carrier is not taken by another user,
\item Users with the lowest utility are given more opportunity to add new carriers to their lists in the next iteration of the algorithm.\\
\ei

The remainder of the paper is organized as follows: The system model is introduced in Section \ref{sec:model}, then the game theoretic formulation is detailed in Section \ref{sec:game-model} and the achievable rates in Section \ref{sec:opt}. The proposed solution is presented in Section \ref{sec:game-solution}, followed by properties of FEAT in Section \ref{sec:theory}. Numerical results are shown in Section \ref{sec:simul}. Finally, concluding remarks are drawn in Section \ref{sec:conc}.

\section{The Multiple Access System Model}
\label{sec:model}
We consider an uplink communication scenario of $K$ parallel Gaussian MAC with $N$ users in the system. These
parallel channels can be treated as different subchannels in a
multi-carrier system. At the transmitter side, the signal
$x_n^k$ transmitted on the $k^{th}$ subcarrier by user $n$ to the base
station (BS) is the normalized (\emph{i.e.}, unit-power) symbol $s_n^{k}$
balanced by some power coefficient $p_n^k: x_n^k = \sqrt{p_n^k} \cdot s_n^k$,
where the symbols $s_n^k$ are usually drawn from Gaussian distribution. At the receiver side,
$y^k$ represents the received signal on the $k^{th}$ channel, which is the sum of the contributions received from all users. Hence, the channel model can we written as
\begin{align}
y^k=\sum_{n=1}^N h_n^k x_n^k + z^k,
\end{align}
where $z^k\sim \mathcal{N}(0,\sigma^2)$ is the additive white Gaussian noise and $h_n^k$ is the fading channel of user $n$ transmitting on carrier $k$ modeled as independent identically distributed (i.i.d.) variables over the Rayleigh fading coefficients. We further assume that the overall bandwidth $W$ can be divided into an arbitrary number of narrow-band carriers ($K\geq2$), the subcarrier is narrow enough to undergo flat fading, and the channel gain is quasi-static fading for which the channels are constant within a given time slot but change independently from one slot to another. Without the constraint of exclusive
assignment of each carrier for users, we generally formulate the problem of
spectral efficiency maximization by allowing that a carrier could be
shared by all users. 
Accordingly, for any user $n \in \{1,2.\ldots, N\}$, the instantaneous per user information rate in bits/s/Hz, also called
\emph{spectral efficiency} is expressed as
\begin{equation}
u_{n} (\textbf{p}) = \sum_{k=1}^{K}\log_2\left(1+\frac{g_{n}^k p_{n}^k}{\sigma^2+ \displaystyle\sum_{\substack{m=1 \\
m\neq n}}^{N} g_{m}^k p_{m}^k}\right),
\label{eq:util_se}\end{equation}
where $g_n^k=|h_n^k|^2$ is the fading channel gain of user $n$ transmitting on carrier $k$. It follows from the above expression that the strategy chosen by a user affects the performance of other users in the
network through multiple-access interference. Now, we can define the optimization problem we are studying as 
\begin{equation}
\{p_{n}^{1*},...,p_{n}^{K*}\} = \argmax{{p_{n}^1,...,p_{n}^K}}
\,u_{n} (\textbf{p}) \label{eq:opt_pb}\end{equation}
subject to
\begin{equation}
\sum_{k=1}^{K} p_n^k = \overline{P_n},
\label{eq:pbudget}\end{equation}
where $\overline{P_n}$ denotes the power budget of user $n$.

\section{The game theoretic formulation}\label{sec:game-model}

The instantaneous information rate determines the
maximum achievable rate over all fading states without a delay constraint. In
this work, we allocate transmit powers to each user (over a total power
budget constraint) in order to maximize the individual transmission rate in (\ref{eq:util_se}). In fact, when
channel state information is made available at the transmitters, users know
their own channel gains and thus they will adapt their transmission strategy
relative to this knowledge. The solution of the optimization problem in (3) is the
well-known \emph{water-filling} allocation \cite{cover} expressed
by\footnote{$(x)^+=\textrm{max}(0,x)$.}:
\begin{equation}
p_n^{k*} =\left(
\frac{1}{\lambda_{0}}-\frac{\sigma^2+\sum_{m\neq n}p_m^k g_m^k }{g_n^k}\right)^+,\\
\label{eq:wf}\end{equation} where $\lambda_{0}$ is
Lagrange's multiplier corresponding to user $n$'s average power constraint in (\ref{eq:pbudget}).

The situation presented above can be described as an $N$-person non-cooperative game with
player utilities $u_n$, $n=1,\ldots,N$ and strategy sets $\mathcal{P}_n=\{ p_n=(p^1_n,\ldots,p^K_n): \sum_{k=1}^{K} p_n^k = \overline{P_n}\}$, $n=1,\ldots,N$.

The basic solution concept used in non-cooperative game theory is that of Nash
equilibrium (NE) \cite{nash50}. It is a vector of strategies (referred to
hereafter and interchangeably as actions) $\textbf{p}^{NE} =
(p_1^{NE},\ldots,p_N^{NE})$, one for each player, such that no player has an
incentive to unilaterally change his strategy, \emph{i.e.},
\begin{equation}
\label{eq:Nash}
u_i(\textbf{p}^{NE})\geq u_i(p_{1}^{NE},\ldots,p_{i-1}^{NE},p_i,p_{i+1}^{NE},\ldots,p_{N}^{NE})\quad\mbox{for every action }p_i \neq
p_i^{NE}. 
\end{equation}

If for some constant $\epsilon>0$ a weaker version of (\ref{eq:Nash}):
$$(1+\epsilon)u_i(\textbf{p}^{NE})\geq u_i(p_{1}^{NE},\ldots,p_{i-1}^{NE},p_i,p_{i+1}^{NE},\ldots,p_{N}^{NE})\quad\mbox{for every action }p_i \neq
p_i^{NE}$$
is satisfied, we say that $\textbf{p}^{NE}$ is an $\epsilon$-Nash equilibrium.

In \cite{Samir13}, it has been shown that the game is an exact potential game and possesses a unique NE in pure strategies. The following lemma describes such equilibrium. 

\begin{lemma} Let the power allocation profiles $p_1=(p^1_1,p^2_1,\ldots,p^K_1) \in \mathcal{P}_1, p_2=(p^1_2,p^2_2,\ldots,p^K_2)\in \mathcal{P}_2, \ldots, p_N=(p^1_N,p^2_N,\ldots,p^K_N) \in \mathcal{P}_N$ form a Nash equilibrium of the game. Then, for all $n \in \lbrace 1,2, \ldots, N \rbrace$ and for all $k \in \lbrace 1,2, \ldots, K \rbrace$, it follows that $p_n^k$ satisfies the equality in \eqref{eq:wf}.
\end{lemma}
\begin{proof}
The proof follows by noticing that for all $n \in \lbrace 1,2, \ldots, N \rbrace$, the power allocation vector $p_n=(p^1_n,p^2_n,\ldots,p^K_n) \in \mathcal{P}_n$, such that for all $k \in \lbrace 1,2, \ldots, K \rbrace$, $p_n^k$ satisfies \eqref{eq:wf} is the best response of player $n$ to the power allocation vectors adopted by all the other players.
\end{proof}

\section{Achievable rates}\label{sec:opt}
In order to gain deeper insights into the achievable rates of the game theoretic scheme, we aim in this section to reveal the relationship between the utility obtained at the NE and the fundamental limit on the individual information transmission rate.

Let us first define the globally optimal solution, which is the set of powers $p_1=(p^1_1,p^2_1,\ldots,p^K_1) \in \mathcal{P}_1, p_2=(p^1_2,p^2_2,\ldots,p^K_2)\in \mathcal{P}_2, \ldots, p_N=(p^1_N,p^2_N,\ldots,p^K_N) \in \mathcal{P}_N$ that maximizes the sum information transmission rate given by

\begin{eqnarray}
\sum_{k=1}^{K} \log_2\left(1+\frac{\displaystyle\sum_{n=1}^{N} g_{n}^k p_{n}^k}{\sigma^2}
\right)=\sum_{n=1}^{N}\sum_{k=1}^{K}\log_2\left(1+\frac{g_{n}^k p_{n}^k}{\sigma^2+ \displaystyle\sum_{\substack{m=1 \\
m< n}}^{N} g_{m}^k p_{m}^k}\right),
\label{eq:util_opt}
\end{eqnarray}

where the equality in (\ref{eq:util_opt}) is proved in the Appendix (see Eq. (\ref{eq:util_unk})).

Then, for all $n \in \lbrace 1,2, \ldots, N \rbrace$,  the information transmission rate achieved by player $n$, denoted by $R_n \geqslant 0$,  depends on the receiver configuration. More specifically, for all $n \in \lbrace 1,2, \ldots, N \rbrace$, the following upper-bound holds in the case of NE, in which the receiver treats interference as noise:
\begin{IEEEeqnarray}{rcl}
R_n & \leq & u_n (\textbf{p}) \\
& = & \sum_{k=1}^{K}\log_2\left(1+\frac{g_{n}^k p_{n}^k}{\sigma^2+ \displaystyle\sum_{\substack{m=1 \\
m\neq n}}^{N} g_{m}^k p_{m}^k}\right).
\label{eq:util-nash}
\end{IEEEeqnarray}
Alternatively, for all $n \in \lbrace 1,2, \ldots, N \rbrace$, the following upper-bound holds in the case of the fundamental limit on the individual information transmission rate for player $n$, for the case of descending decoding order, see for instance \cite{cover}. Here, the receiver implements a successive interference cancellation (SIC) scheme with descending decoding order:
\begin{IEEEeqnarray}{rcl}
\label{EqFLbound}
R_n & \leq & \sum_{k=1}^{K} \log_2 \left(1+\frac{g_{n}^k p_{n}^k}{\sigma^2+ \displaystyle\sum_{\substack{m=1 \\m<n}}^N g_{m}^k p_{m}^k}\right).
\label{eq:util-opt}
\end{IEEEeqnarray}

Moreover, note also that, for all $n \in \lbrace 1,2, \ldots, N \rbrace$ and for all $\textbf{p} \in \mathcal{P}_1 \times \mathcal{P}_2\times \ldots \times \mathcal{P}_N$,
\begin{IEEEeqnarray}{rcl}\label{eq:util_ub}
u_n (\textbf{p}) =  \sum_{k=1}^{K}\log_2\left(1+\frac{g_{n}^k p_{n}^k}{\sigma^2+ \displaystyle\sum_{\substack{m=1 \\
m\neq n}}^{N} g_{m}^k p_{m}^k}\right) & \leq & \sum_{k=1}^{K} \log_2 \left(1+\frac{g_{n}^k p_{n}^k}{\sigma^2+ \displaystyle\sum_{\substack{m=1 \\m<n}}^N g_{m}^k p_{m}^k}\right).
\label{eq:inequn}
\end{IEEEeqnarray}
The inequality in (\ref{eq:inequn}) is proved in the Appendix (see Eq. (\ref{eq:proof_util_1})).
In particular, when the power profile $\textbf{p}$ in Eq. (\ref{eq:inequn}) satisfies the equality in Eq. (\ref{eq:wf}), the above inequality sets an upper-bound for the individual utility of player $n$, which turns out to be relevant for the performance analysis of the algorithm FEAT later in the paper.

\section{Approximate solutions for the $N$-player game}\label{sec:game-solution}

This section introduces the main contribution, namely, an algorithm whose output is the set of channel indices assigned to each transmitter. This channel assignment is performed taking into account the following observations: \newline
$(a)$ When all transmitters are let to use all the available channels, the highest information sum rate is achieved by a scheme in which transmitters are assigned an order. Transmitter $1$ uses a power allocation obtained by using the water-filling algorithm over all channels in the absence of interference. 
For all $i \in \lbrace 2,3, \ldots, n \rbrace$, Transmitter $i$ uses a power allocation obtained by using the water-filling algorithm considering the interference of all previous transmitters. At the Receiver, a SIC algorithm is implemented in which transmitters are decoded in descending order. That is, Transmitter $n$ is decoded first, whereas, Transmitter $1$ is decoded in the last position. The highest information sum rate is independent of the order given to the transmitters, under the assumption that SIC is based on perfect decoding. Alternatively, the individual information rate does depend on such order \cite{Viswanath}.\newline
$(b)$ In the absence of SIC, all transmitters are subject to mutual interference and thus, the sum information rate is severely decreased. In this context, constraining the transmitters to use only a subset of all available channels have been shown to be an interesting alternative \cite{Perlaza-Crowncom-2009}. \newline
$(c)$ When transmitters are constrained to use only a subset of all available channels, e.g., $L$ channels ($L<K$), the individual rates depend on the choice of such $L$ channels for each transmitter. More specifically, given an order $g_i^{\ell_1} > g_i^{\ell_2} > \ldots > g_i^{\ell_k}$, for some $i \in \lbrace 1,2, \ldots, n \rbrace$ such that for all $s \in \lbrace 1,2, \ldots, k \rbrace$, $\ell_s \in 
\lbrace 1,2, \ldots, k \rbrace$, a channel assignment consisting in the set $\lbrace g_i^{\ell_1}$, $g_i^{\ell_2}$, $\ldots$, $g_i^{\ell_L} \rbrace$ leads to a higher individual information rate for Transmitter $i$ than any other channel assignment. That is, individual rates are monotonically increasing with the channel gains. \newline
$(d)$ When the ratio between number of channels and number of transmitters is small, it is likely that the best channel assignment for Transmitter $i$ and Transmitter $j$ are sets whose intersection is not empty, which leads to conflict. 
 
The observations above lead to the notion of fairness when transmitters are subject to use only disjoints subsets of channels. From observations $(c)$ and $(d)$, it follows that assigning always the best channels to the same user is far from fair. 
The proposed algorithm is based on the following arbitrary assumption: When a channel must be assigned to a transmitter, the benefit of such assignment can be measured by the ratio between the channel coefficient of the assigned channel, and the highest channel coefficient available for such transmitter. That is, the higher this ratio, the more beneficial the allocation is. \\\\

Suppose the assignment of one channel is done according to the function $\nu:\{ 1,\ldots,N\}\rightarrow\{ 1,\ldots,K\}$. The global benefit of this assignment can be measured by the quantity 
\beq
\alpha(\nu):=\min_{n\leq N}\frac{g^{\nu(n)}_n}{\max_{l\leq K}g^l_n}.
\eeq

The main goal of the FEAT algorithm is to find channel assignments that maximize the value of $\alpha$ every time that a channel must be assigned.

The proposed algorithm assigns a different channel to all users at each round. This is done in such a way that $\alpha$ is maximized at the end of the round.

The FEAT algorithm has two parameters $\delta>0$ and $\beta\in(0,1)$: $\delta$ is maximum error of computing $\alpha$ so it should be small, whereas $\beta$ is the percentage of the maximum throughput below which we assume that a user is not served well enough, thus he gets additional carriers to transmit on; it should then be relatively big. For a vector $\pi$, $|\pi|$ denotes its length.

\begin{algorithm}{ \bf{FEAT : Fair coordinatEd iterAtive waTer-filling algorithm}}\\
\label{alg:alpha_choice}
We start with $M=K$, $\pi^*=[1,\ldots,N]$, $\pi_0=[1,\ldots,N]$ and $\rho_n^k=\frac{g_n^k}{\max_{l\leq K}g_n^l}$, $n=1,\ldots,N$, $k=1,\ldots,K$ and create empty lists $\mathcal{L}_n$, $n=1,\ldots,N$.\\
While $M>0$ do\\
{\bf Phase A:} For $n=1,\ldots,|\pi^*|$, sort $\rho_{\pi^*(n)}^k$, $k=1,\ldots,K$ from the biggest to the smallest, obtaining $\rho_{\pi^*(n)}(1),\ldots,\rho_{\pi^*(n)}(K)$.
Set $\underline{\alpha}=0$, $\overline{\alpha}=1$ and $\alpha^*=0$.\\
Repeat the following steps until the procedure is interrupted in point A.1):
\begin{enumerate}[{A}.1)]
\item If $\overline{\alpha}-\underline{\alpha}<\delta$ stop. Otherwise take $\alpha^*=\frac{\overline{\alpha}+\underline{\alpha}}{2}$ and $\pi=\boldsymbol{0}_{1\times K}$.
\item For $n=1,\ldots,|\pi^*|$ do the following steps:\\
Find the smallest $\rho_{\pi^*(n)}(l^*)$ such that $\rho_{\pi^*(n)}(l^*)\geq\alpha^*$. If $\pi(l^*)=0$ then put $\pi(l^*)=\pi^*(n)$. Otherwise, find the biggest $l<l^*$ such that $\pi(l)=0$, and put $\pi(l)=\pi^*(n)$. If $\pi(l)\neq 0$ for every positive $l<l^*$ put $\overline{\alpha}=\alpha^*$ and return to point A.1).
\item Put $\pi^*=\pi$ with all the zero elements removed, $\underline{\alpha}=\alpha^*$ and return to point A.1).
\end{enumerate}
{\bf Phase B:} $m_0=0$. For $n=1,\ldots,|\pi^*|$, do the following steps:
\begin{enumerate}[{B}.1)]
\item $k$ with the highest $\rho^k_{\pi^*(n)}$ is chosen. Let it be denoted by $\mathcal{K}(n)$.
\item If $\sum_{l\in\mathcal{L}_{\pi^*(n)}}\frac{1}{g^l_{\pi^*(n)}}>\frac{|\mathcal{L}_{\pi^*(n)}|}{g^{\mathcal{K}(n)}_{\pi^*(n)}}-\frac{\overline{P}_{\pi^*(n)}}{\sigma^2}$, then $\mathcal{K}(n)$ is added to the list $\mathcal{L}_{\pi^*(n)}$, $m_0=m_0+1$ and
for $i=1,\ldots,N$, $\rho^{\mathcal{K}(n)}_i=0$. Else $\pi^*(n)$ is removed from the list $\pi_0$.
\end{enumerate}
{\bf Phase C:} Set $M=M-m_0$ and $\pi^*=[]$ (empty vector).\\
If $|\pi_0|>0$ and $M>0$, do the following steps:
\begin{enumerate}[{C}.1)]
\item For $n=1,\ldots,|\pi_0|$ compute
$Q(n)=\sum_{k\in\mathcal{L}_{\pi_0(n)}}\log_2\left[ \frac{g^k_{\pi_0(n)}}{|\mathcal{L}_{\pi_0(n)}|}\left( \frac{\overline{P}_{\pi_0(n)}}{\sigma^2}+\sum_{l\in\mathcal{L}_{\pi_0(n)}}\frac{1}{g^l_{\pi_0(n)}}\right)\right]$.\\
Put $Q_{\max}=\max_{n=1,\ldots,|\pi_0|}Q(n)$.
\item Sort pairs $(\pi_0(n),Q(n))$, $n=1,\ldots,|\pi_0|$ by their second coordinates in an increasing order, obtaining $[\pi^*,Q^*]$.\\
If $Q^*(1)\leq\beta Q_{\max}$ and $m_0>0$, take $k^*=\max\{ k\leq\min\{ M,|\pi_0|\}: Q^*(k)\leq\beta Q_{\max}\}$, $\pi^*=\pi^*(1,\ldots,k^*)$,\\
else $\pi^*=\pi^*(1,\ldots,\min\{ M,|\pi_0|\})$;
\end{enumerate}
else break.
\end{algorithm}

The output of the algorithm is the partition of the spectrum $\{ 1,\ldots,K\}$ into lists $\mathcal{L}_1,\ldots,\mathcal{L}_N$. Then the strategy for each player $n$ is obtained by applying water-filling from Eq. (\ref{eq:wf}) only to carriers from his list $\mathcal{L}_n$.

To understand what we do in the algorithm presented above on an intuitive level, first note that in each iteration of its main loop, at most one carrier is added to the list of each player. Below, we will try to explain the way these additions are made phase by phase.

First, to understand the sense of Phase A, note that the fraction $\rho_n^k$ appearing there can be interpreted as a measure of disutility of player $n$ from choosing carrier $k$ instead of his best one. In the algorithm, we aim at dividing the spectrum into separate sets of carriers with only one user transmitting on each. So the utilities from using a given carrier are always of the form $\log_2\left( 1+\frac{pg_n^k}{\sigma^2}\right)$ (for a while we assume that the transmission power $p$ is fixed, rather than chosen by water-filling). Maximizing the value of $\rho_{ik}$ is thus equivalent to choosing the carrier with highest utility.
Given this interpretation, the $\alpha^*$ appearing in Phase A of the algorithm can be interpreted as the maximal disutility for any player from not choosing his carrier first, that is the worst-case\footnote{Worst-case here means that such a big disutility will only be possible if different users' private ordering (from best to worst) of the channels is similar.} ratio of utility of any of the players who do not choose their carriers first to their utility if they were the first ones to choose. The sense of Phase A is thus finding the ordering of the players which minimizes this disutility.
It is done by putting on $i$-th coordinate of ordering $\pi^*$ a player (his index), who has at least $i$ good channels to choose from (by which we mean $i$ channels with utility better than $\alpha^*$ times his best possible utility if he was a leader). $\alpha^*$ found by FEAT is the maximal value (computed with a $\delta$ toleration) for which such an ordering is possible.

Next, in Phase B, players ordered in Phase A choose one by one the best carrier (the one that gives them the highest utility) which is not yet taken by them or one of the other users and add these carriers to the lists of carriers where they are allowed to transmit.

Finally, in Phase C, the utility of each user when water-filling is applied to their current lists of carriers, $Q(n)$ is computed. The players with the lowest utility (precisely -- those whose utility is smaller than $\beta$ times the biggest utility among the users) are chosen and they are given the opportunity to profit from additional carriers to their lists in the next iteration of the main loop of the algorithm.

In summary, we have the following information to be exchanged between the base station and users: 
\be
\item Each UE $n$ sends beacons to his serving BS in order to estimate his channel gain at the $k^{th}$ subcarrier $g_n^k$,
\item BS computes and sends the lists $\mathcal{L}_n$ for each user $n$,
\item Each UE $n$ computes his power policy using Eq. (3) by applying the water-filling only to carriers from his list $\mathcal{L}_n$. 
\ee

\section{Properties of FEAT}\label{sec:theory}

\subsection{Disjoint carrier water-filling game}
In this section, we consider a slightly modified water-filling game, assuming that, as before the goal of each user is to maximize his throughput, but no two users are allowed to transmit on the same carrier. Note that such a game is a generalized game in the sense of Debreu \cite{Debreu}. Equilibria in such a game are defined as in classical noncooperative games with a slight change that the set of actions available to a player depend on those applied by the others. We will show that an equilibrium in such a modified water-filling game (called disjoint carrier water-filling game in the sequel) is what we obtain as a result of FEAT. We begin the section with a simple characterization of the set of all equilibria of this game.

\begin{prop}
\label{prop:DCWGequilibria}
Suppose the vector $\textbf{p}=(p_1,\ldots,p_N)$ is such that $\mathcal{L}_i:=\{ k\in\{ 1,\ldots, K\}: p_i^k>0\}$, $i=1,\ldots,N$ are disjoint subsets of $\{ 1,\ldots, K\}$ and that $\mathcal{L}_0:=\{ 1,\ldots,K\}\setminus\bigcup_{i=1}^N\mathcal{L}_i$. Then, $\textbf{p}$ is a NE in the disjoint carrier water-filling game, if and only if for each user $i$, $p_i$ is his water-filling allocation with the set of available carriers reduced to $\mathcal{L}_i\cup\mathcal{L}_0$.
\end{prop}

\begin{proof} We know that the power allocation maximizing throughput of player $i$ is water-filling allocation \cite{cover}, thus in any NE of the game each player has to waterfill on all the available carriers. Since in the situation when other users use allocations $p_j$, $j\neq i$, he is limited to the set $\mathcal{L}_i\cup\mathcal{L}_0$ by the rules of the game, $\textbf{p}$ can only be a NE iff $i$ is water-filling on this set and if the same is true for any other user.
\end{proof}

An immediate consequence of Proposition \ref{prop:DCWGequilibria} is as follows:
\begin{corol}
\label{cor:DCWGalgoNE}
The output of FEAT corresponds to a NE of the disjoint carrier water-filling game.
\end{corol}

\begin{proof} Note that whenever the output of FEAT is a partition of entire set of carriers $\{ 1\ldots, K\}$ into disjoint sets of carriers of each player $\mathcal{L}_i$, $i=1,\ldots,N$, water-filling on these sets will clearly be a NE in the disjoint water-filling game by Proposition \ref{prop:DCWGequilibria}. Suppose then that the set $\mathcal{L}_0=\{ 1,\ldots,K\}\setminus\bigcup_{i=1}^N\mathcal{L}_i$ is nonempty. In that case, FEAT can only terminate, if $|\pi_0|=0$. But this means that all the users have been removed from $\pi_0$, because there were no carriers left such that transmission on them could improve their utility from water-filling. Hence, the water-filling allocation of any player $i$, if he waterfilled on $\mathcal{L}_i$ would be the same as if he waterfilled on $\mathcal{L}_i\cup\mathcal{L}_0$, which implies that the output of the algorithm corresponds to an equilibrium in the disjoint carrier water-filling game.
\end{proof}

\subsection{General properties of the algorithm}
In the next proposition we are giving the worst-case bounds on crucial metrics by which we can evaluate the output of FEAT as a solution to the water-filling game.

\begin{prop}
\label{prop:algo_properties}
Suppose that $N\leq K$ and that the sets $\mathcal{L}_1,\ldots,\mathcal{L}_N$ are outputs of FEAT and that $\textbf{p}=(p_1,\ldots,p_N)$ is such that for each $i$, $p_i$ is a water-filling allocation of user $i$ limited to carriers from $\mathcal{L}_i$. Let further $\alpha^*_1$ denote the value of $\alpha^*$ at the end of the first iteration of the main loop of FEAT and define the following quantities:
$$\Theta_{\max}=\max_{i\in\{ 1,\ldots,N\}}\max_{k\in\{ 1,\ldots,K\}}\overline{P}_ig_i^k,$$
$$\Theta_{\min}=\min_{i\in\{ 1,\ldots,N\}}\max_{k\in\{ 1,\ldots,K\}}\overline{P}_ig_i^k,$$
$$\Omega=\frac{\Theta_{\max}}{\sigma^2\alpha^*_1\ln\left( 1+\frac{\Theta_{\max}}{\sigma^2}\right)}+\frac{1-\alpha^*_1}{(\alpha^*_1)^2\ln\left( 1+\frac{\Theta_{\min}}{\sigma^2}\right)}>1.$$
Then:
\begin{enumerate}[(a)]
\item $\textbf{p}$ is an $(\Omega-1)$-NE in the water-filling game.
\item The social welfare in the water-filling game when the strategy vector $\textbf{p}$ is used, is at most $\Omega$ times smaller than its optimal value.
\item If all the users apply strategies $\textbf{p}$, the ratio $\frac{\max_{i\in\{ 1,\ldots,N\}}u_i}{\min_{i\in\{ 1,\ldots,N\}}u_i}$ is bounded above by $\frac{\Theta_{\max}}{\sigma^2\log_2\left( 1+\frac{\alpha^*_1\Theta_{\min}}{\sigma^2}\right)}$.
\end{enumerate}

\end{prop}

\begin{proof}
(a) The biggest individual loss of utility by a player that can happen as a consequence of using FEAT is when a player, say player $i$, is forced by the algorithm to use only his worst carrier $k^*$ (the one with the lowest $g_i^k$), while he should use all of them to maximize his utility. Moreover, the interference on these other carriers is negligible. To bound this loss of utility, we will assume there is no interference at all on any carrier. Then, since the worst carrier can only be chosen by FEAT as the only carrier to transmit on if $\rho_i(k^*)=\frac{g_i^{k^*}}{\max_{k\in\{ 1,\ldots,K\}}g_i^k}\geq\alpha^*_1$, to maximize player $i$'s utility we should assume that $g_i^k=\frac{g_i^{k^*}}{\alpha_1^*}$ for $k\neq k^*$. Then, the utility of player $i$ from using FEAT is
\begin{equation}
\label{ui0}
u_i^0=\log_2\left( 1+\frac{\overline{P}_ig_i^{k^*}}{\sigma^2}\right)
\end{equation}
while the upper bound on his utility $u_i^{\max}$ can be computed as follows:
From the water-filling formula (\ref{eq:wf}) we know that:
$$p_i^{k^*}=\left( \frac{1}{\lambda_0^i}-\frac{\sigma^2}{g_i^{k^*}}\right)^+,$$
$$p_i^{k}=\left( \frac{1}{\lambda_0^i}-\frac{\alpha^*_1\sigma^2}{g_i^{k^*}}\right)^+\mbox{ for }k\neq k^*$$
with $\lambda_0^i$ satisfying
$$\frac{K}{\lambda_0^i}-\frac{\sigma^2}{g_i^{k^*}}(1+\alpha_1^*(K-1))=\overline{P}_i.$$

After some algebra, we thus obtain
$$p_i^{k^*}=\frac{\overline{P}_i}{K}-\frac{(K-1)\sigma^2}{g_i^{k^*}K}(1-\alpha_1^*),$$
$$p_i^{k}=\frac{\overline{P}_i}{K}+\frac{\sigma^2}{g_i^{k^*}K}(1-\alpha_1^*)\mbox{ for }k\neq k^*$$
if the former is positive or 
$$p_i^{k^*}=0,\quad p_i^{k}=\frac{\overline{P}_i}{K-1}\mbox{ for }k\neq k^*$$
otherwise. We shall concentrate on the first case -- all the bounds obtained will trivially also hold for the second one.

Inputting the powers into utility, we get
\begin{eqnarray*}
u_i^{\max}&=&(K-1)\log_2\left( 1+\frac{\overline{P}_ig_i^{k^*}}{\alpha^*_1\sigma^2K}+\frac{1-\alpha^*_1}{\alpha^*_1K}\right) +\log_2\left( 1+\frac{\overline{P}_ig_i^{k^*}}{\sigma^2K}-\frac{(K-1)(1-\alpha_1^*)}{K}\right)\\
&<&K\log_2\left( 1+\frac{\overline{P}_ig_i^{k^*}}{\alpha^*_1\sigma^2K}+\frac{1-\alpha^*_1}{\alpha^*_1K}\right).
\end{eqnarray*}
Note that for any $A>0$
\begin{equation}
\label{diffx1}
\left(\frac{\log_2(1+Ax)}{x}\right)_x=\frac{\frac{Ax\ln(2)}{1+Ax}-\log_2(1+Ax)}{x^2}.
\end{equation}
As $\left( \frac{Ax\ln(2)}{1+Ax}-\log_2(1+Ax)\right)_x=\ln(2)\left( \frac{A}{(1+Ax)^2}-\frac{A}{1+Ax}\right)=\frac{-A^2\ln(2)}{(1+Ax)^2}<0$, the RHS of (\ref{diffx1}) is decreasing in $x$ on $(0,+\infty)$, which means that its biggest value of 0 is obtained for $x=0$.
Hence, for any positive constant $A$, $\frac{\log_2(1+Ax)}{x}$ is a decreasing function of $x$ for $x>0$ and $\lim_{x\rightarrow 0}\frac{\log_2(1+Ax)}{x}=\frac{A}{\ln 2}$. Applying this to the obtained bound on $u_i^{\max}$ we obtain
$$u_i^{\max}<\frac{1}{\ln 2}\left(\frac{\overline{P}_ig_i^{k^*}}{\alpha^*_1\sigma^2}+\frac{1-\alpha^*_1}{\alpha^*_1}\right).$$
Using this and (\ref{ui0}) we thus obtain
\begin{equation}
\frac{u_i^{\max}}{u_i^0}<\frac{\frac{\overline{P}_ig_i^{k^*}}{\alpha^*_1\sigma^2}+\frac{1-\alpha^*_1}{\alpha^*_1}}{\ln\left( 1+\frac{\overline{P}_ig_i^{k^*}}{\sigma^2}\right)}
= \frac{\frac{\overline{P}_ig_i^{k^*}}{\sigma^2}}{\alpha^*_1\ln\left( 1+\frac{\overline{P}_ig_i^{k^*}}{\sigma^2}\right)}+\frac{1-\alpha^*_1}{\alpha^*_1\ln\left( 1+\frac{\overline{P}_ig_i^{k^*}}{\sigma^2}\right)}
\leq\Omega,
\label{eq:utility_loss}
\end{equation}
which implies that $\textbf{p}$ is an $(\Omega-1)$-NE in the water-filling game.

(b) If we use the inequality (\ref{eq:utility_loss}) for each player, we obtain that the ratio of the upper bound on the social welfare and the social welfare when strategy vector $\textbf{p}$ is applied can be bounded as follows:
$$\frac{\sum_{i=1}^N u_i^{\max}}{\sum_{i=1}^Nu_i^0}\leq \frac{\sum_{i=1}^N \Omega u_i^0}{\sum_{i=1}^Nu_i^0}=\Omega,$$
which ends the proof of part (b) of the proposition.

(c) We know that after the first iteration of the main loop of FEAT, each user $i$ gets one carrier, say $k_i$, to transmit on, such that 
$$\rho_i(k_i)=\frac{g_i^{k_i}}{\max_{\{ k\in\{ 1,\ldots, K\}}g_i^k}\geq \alpha^*_1.$$
This obviously implies that his utility corresponding to the vector of strategies $\textbf{p}$ (that is -- at the end of the algorithm), is not smaller than
\begin{equation*}
\label{eq:denominator}
u_i^{\min}=\log_2\left( 1+\frac{\overline{P}_i\alpha_1^*\max_{\{ k\in\{ 1,\ldots, K\}}g_i^k}{\sigma^2}\right)\geq \log_2\left( 1+\frac{\Theta_{\min}}{\sigma^2}\right).
\end{equation*}
On the other hand, the maximal utility that any user $i$ can obtain when he gets no interference from other users is (here the vector $p^{WF}_i$ denotes the water-filling allocation of player $i$)
\begin{eqnarray*}
u_i^{\mbox{WF}}&=&\sum_{k=1}^K\log_2\left( 1+\frac{p^{WF,k}_ig_i^k}{\sigma^2}\right)<\sum_{k=1}^K\frac{p^{WF,k}_ig_i^k}{\sigma^2}\nonumber\\
&\leq& \sum_{k=1}^K\frac{p^{WF,k}_i\max_{l\in\{ 1,\ldots, K\}}g_i^l}{\sigma^2}=\frac{\overline{P}_i\max_{l\in\{ 1,\ldots, K\}}g_i^l}{\sigma^2}\leq\frac{\Theta_{\max}}{\sigma^2}.\label{eq:numerator}
\end{eqnarray*}
Dividing (\ref{eq:numerator}) side by side by (\ref{eq:denominator}), we obtain the desired inequality.
\end{proof}

\begin{remark}

The value of $\Omega$ does not directly depend on $K$, so it suggests that the quality of the solution obtained by FEAT (by which we may mean the `distance' from the social optimum or individually optimal solution described by a NE) does not depend on $K$ at all. Note however that the constants $\alpha^*_1$, $\Theta_{\max}$ and $\Theta_{\min}$ all depend on the values of $g_i^k$s which are in fact i.i.d. exponential variables. From simulations we can infer that $\Omega$ depends on $K$ sublinearly.

\end{remark}

\begin{remark}
\label{rem:top}

The bounds given by parts (a) and (b) of Proposition \ref{prop:algo_properties} are tight in the sense that we can get arbitrarily close to these bounds when $K$ is large and channels are strongly correlated across users. Moreover, the qualities of most carriers are not good, so at least one of the users ends up assigned to some of his bad carriers whatever the value of $K$ is.

As for the bound given in part (c) of the proposition, it only takes into account what happens during the first iteration of the main loop of FEAT, so the bound is rather loose. Note, that one of the main goals of the algorithm is to minimize the ratio between the best and the worst utilities of the users. The way it is done is however rather difficult to quantify.
\end{remark}

\begin{remark}
\label{rem:last}
Note that the situation described in Remark \ref{rem:top}, which is used to prove (a) and (b) of Proposition \ref{prop:algo_properties}, although theoretically possible, is rather unlikely to happen if the CQIs are i.i.d. exponential variables. In fact, what we should expect in a typical case is that the second term in the bound of $\frac{u_i^{\max}}{u_i^0}$ given by (\ref{eq:utility_loss}) will disappear when $K$ becomes large, as $\alpha^*_1 \rightarrow 1$ and the denominator of this part will grow with $K$ going to infinity.
On the other hand, the denominator of the first term when $K>N$ can be replaced by the total utility from using all the carriers assigned to user $i$ by FEAT, which will grow at a similar rate as the numerator with increase of $K$ (as additional carriers are added to $\mathcal{L}_i$).
We are thus convinced, that the expected quality of the solutions found by FEAT will in fact tend to a constant for any given $N$ as $K\rightarrow\infty$, which is much better than the Proposition \ref{prop:algo_properties} suggests. Quantifying these expected values is again however much more difficult than obtaining the worst-case estimates given there.

\end{remark}

\subsection{Complexity of the algorithm}
It is easy to see that the complexity of Phase A of FEAT is $O(NK\log(K))$, that of Phase B is (note that the number of iterations of the loop inside it is bounded by a number depending only on $\delta$) $O(NK)$, while that of Phase C is $O(N\log(N))$. Since on each iteration of the main loop of the algorithm at least one carrier is assigned to a user or (in case none were assigned) the list of users who are still available for further assignments is shrunk, the total number of iterations of the main loop of the algorithm is bounded above by $N+K$. Finally, note that by assumption, FEAT is applied only in case $N\leq K$. Putting all this information together, we can conclude the following:

\begin{prop}
\label{prop:complexity}
The complexity of FEAT is $O(NK^2\log(K))$.
\end{prop}

\section{Numerical results}\label{sec:simul}
In this section, simulation results are presented to verify the performances of the proposed algorithm FEAT.
For comparison purpose, we exemplify our general analysis by investigating the following algorithms:
\bi
\item \emph{\textbf{the Nash strategy}}: which is the strategy of the non-cooperative game theoretic problem whose solution is defined in Section \ref{sec:game-model}. Then, the utility is computed by using (\ref{eq:util_se}), treating the interference as noise at the receiver side.
\item \emph{\textbf{the optimal strategy}}: which corresponds to the fundamental limit on the individual information transmission rate obtained when all users waterfill considering all other users as interference as expressed in Eq. (\ref{eq:wf}). Then, the utility is computed using the right hand side of (\ref{eq:util-opt}) by applying SIC at the receiver side, \emph{i.e.}, by considering only predecessors as interference (see Section \ref{sec:opt}). 

This will thus serve as the
optimal social welfare solution for the sum rate maximization problem in order to demonstrate how much gain may theoretically be exploited through considering such a global optimal solution with respect to the other schemes. It is noteworthy that, as already mentioned in Section \ref{sec:opt}, the optimal solution requires to implement a SIC scheme at the receiver side, whereas Nash and FEAT do not. 
\item \emph{\textbf{the spectrum pooling strategy}}: which is throughput-based-utility \cite{Jondral_pooling,MajedIET08}. This is done by overlaying users into a common pool. The spectrum pooling behavior is assumed to allow only one user to
simultaneously transmit over the same sub-band.
\ei

\section*{Transmission rate}

\begin{figure}[t]
\centering
\vspace*{-4.5cm}
\hspace*{-0cm}
\includegraphics[height = 13.5cm,width=11cm]{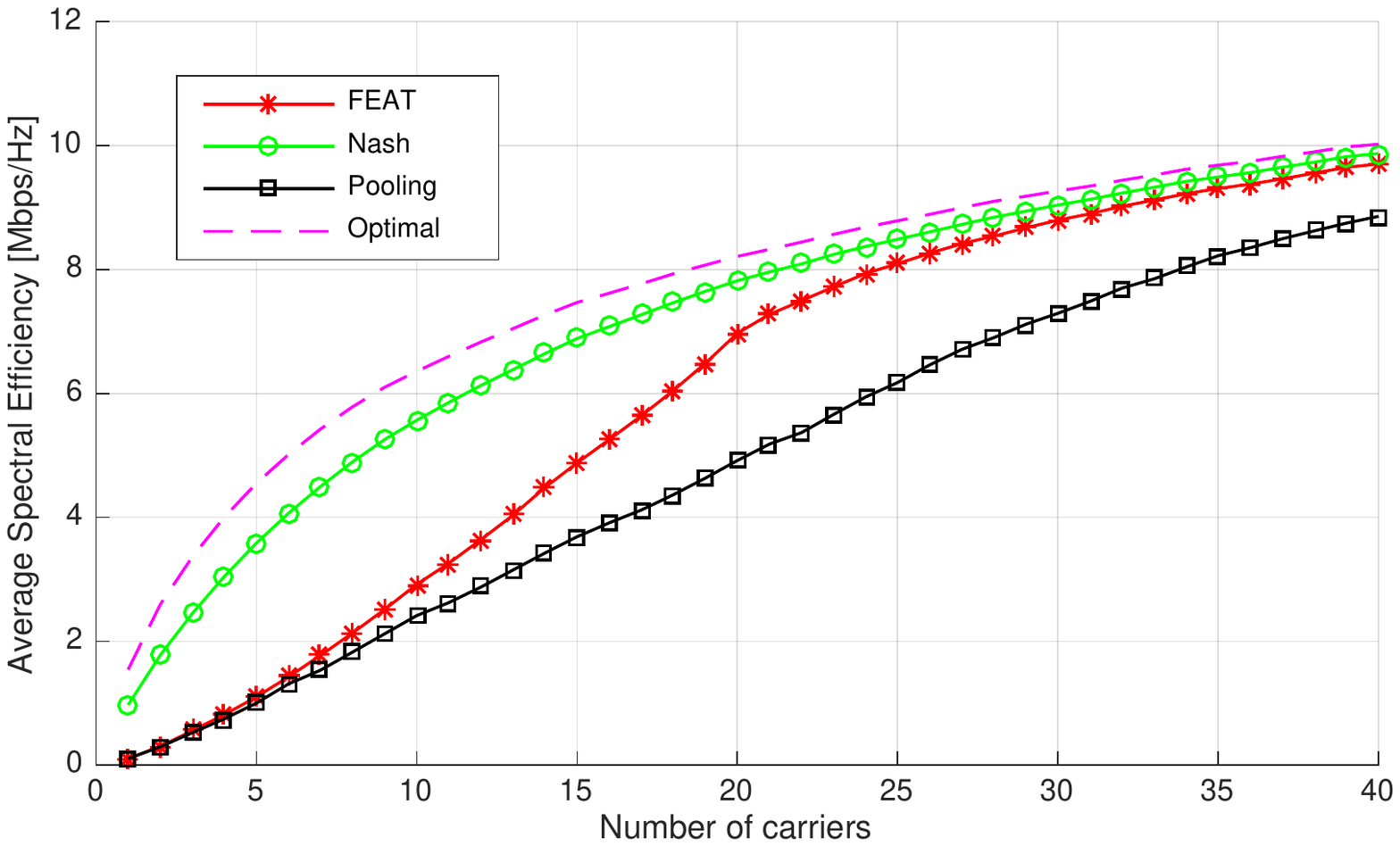}
\vspace{-4.5cm}
\caption{The average utility as function of $K$ for SNR $=-10$ dB and $N=20$.}
\label{fig:SW_K_algo_-10dB}
 \end{figure}

\begin{figure}[t]
\vspace*{-4.5cm}
\centering
\hspace*{-0cm}
\includegraphics[height = 13.5cm,width=11cm]{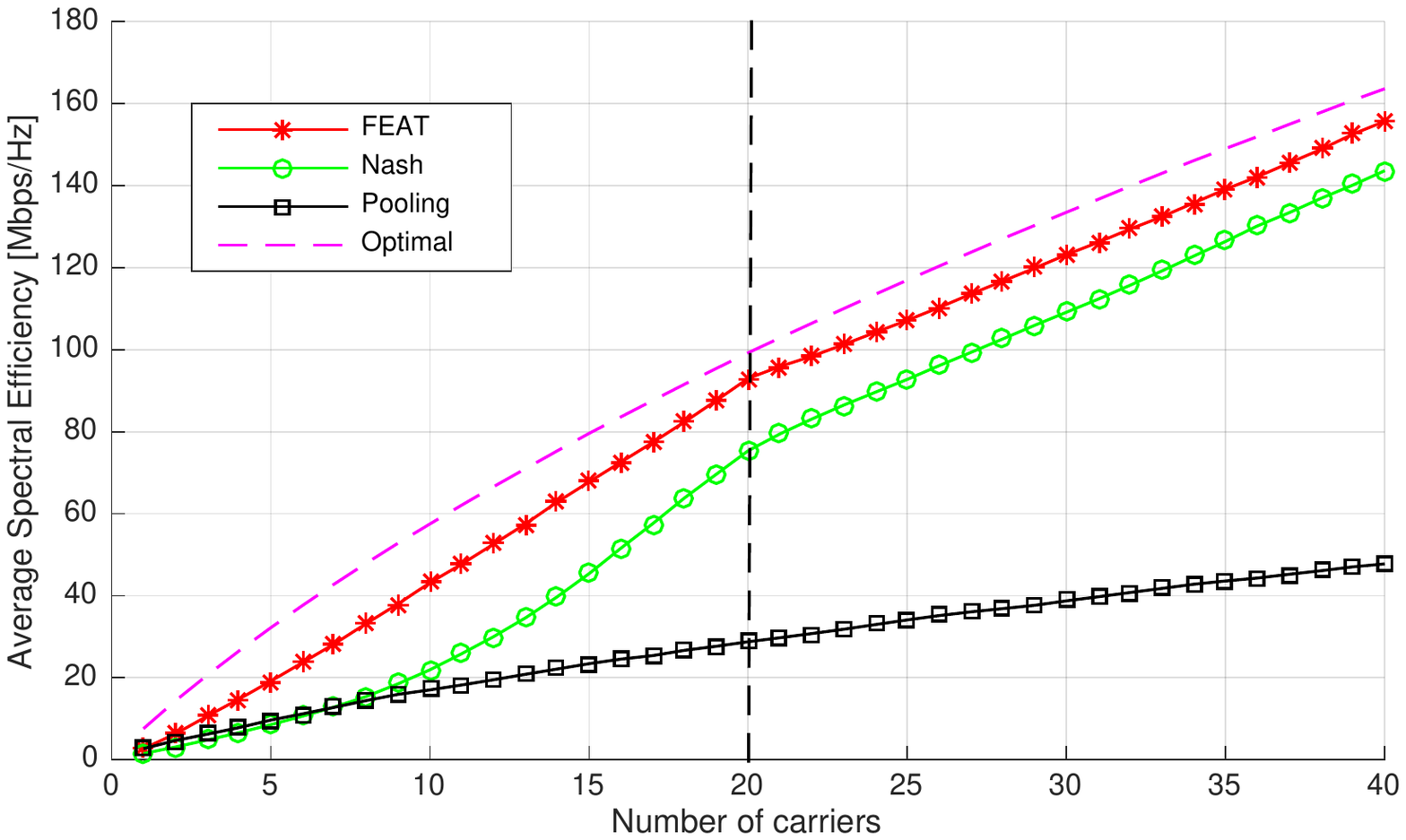}
\vspace*{-4.5cm}
\caption{The average utility as function of $K$ for SNR $=10$ dB and $N=20$.}
\label{fig:SW_K_algo_N20_10dB}
 \end{figure}

\begin{figure}[t]
\centering
\vspace*{-3cm}
\hspace*{-0cm}
\includegraphics[height = 11cm,width=11cm]{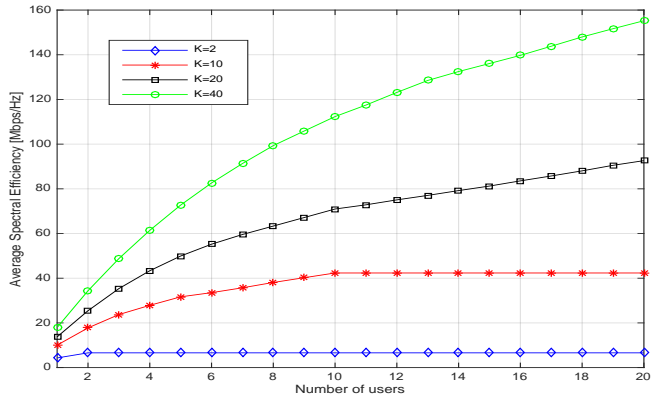}
\vspace{-3cm}
\caption{The average utility as function of $N$ for different $K$ at SNR = $10$ dB.}
\label{fig:SW_N_K}
\end{figure}

\begin{figure}[t]
\centering
\vspace*{-3cm}
\begin{subfigure}
\centering
\hspace*{-0cm}
\hspace*{-0cm}
\includegraphics[height = 11cm,width=11cm]{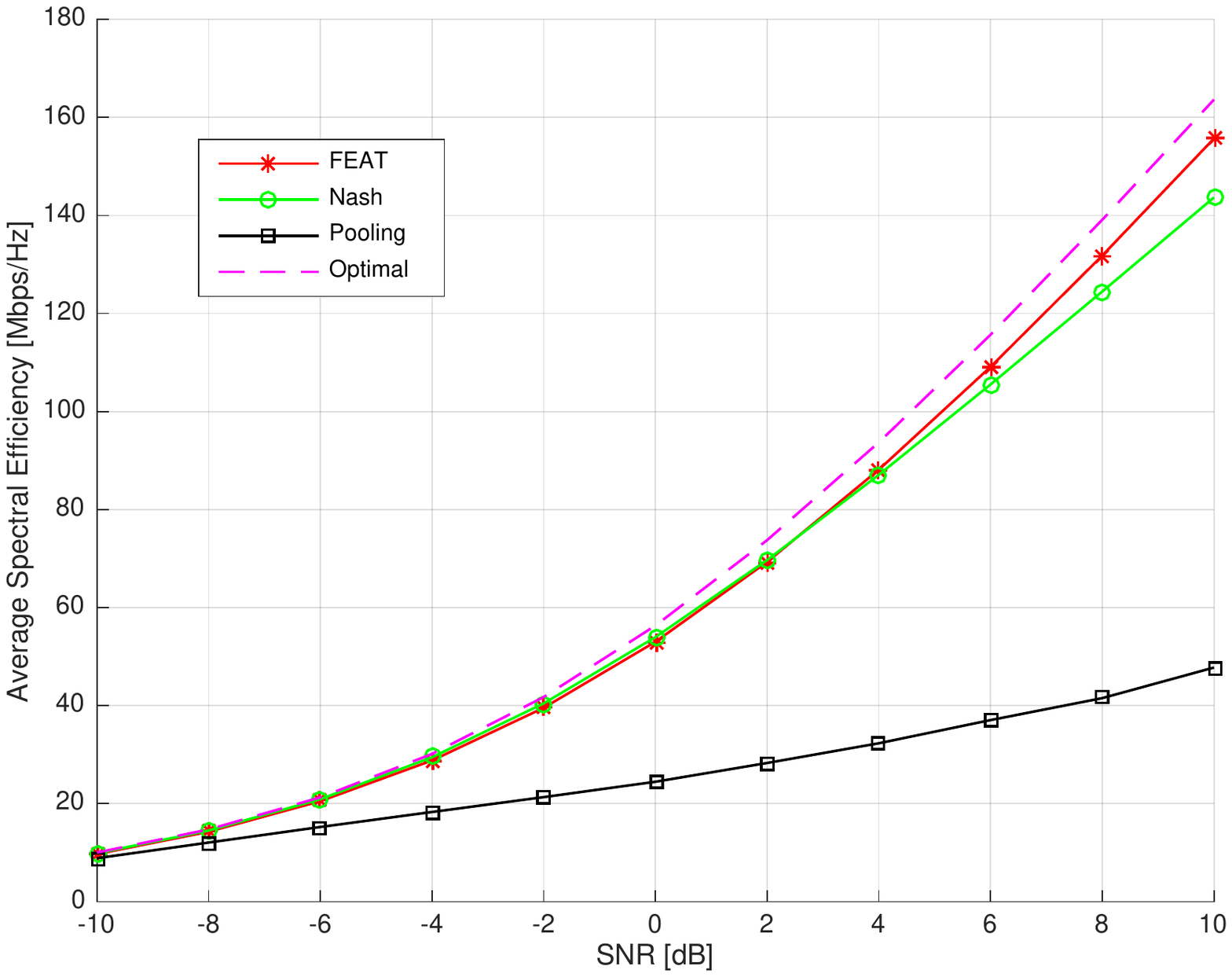}
\vspace{-3cm}
\caption{The average utility as function of SNR for $K=40$ and $N=20$.}
\label{fig:SW_SNR_algo}

\vspace*{-2cm}
\includegraphics[height = 11cm,width=11cm]{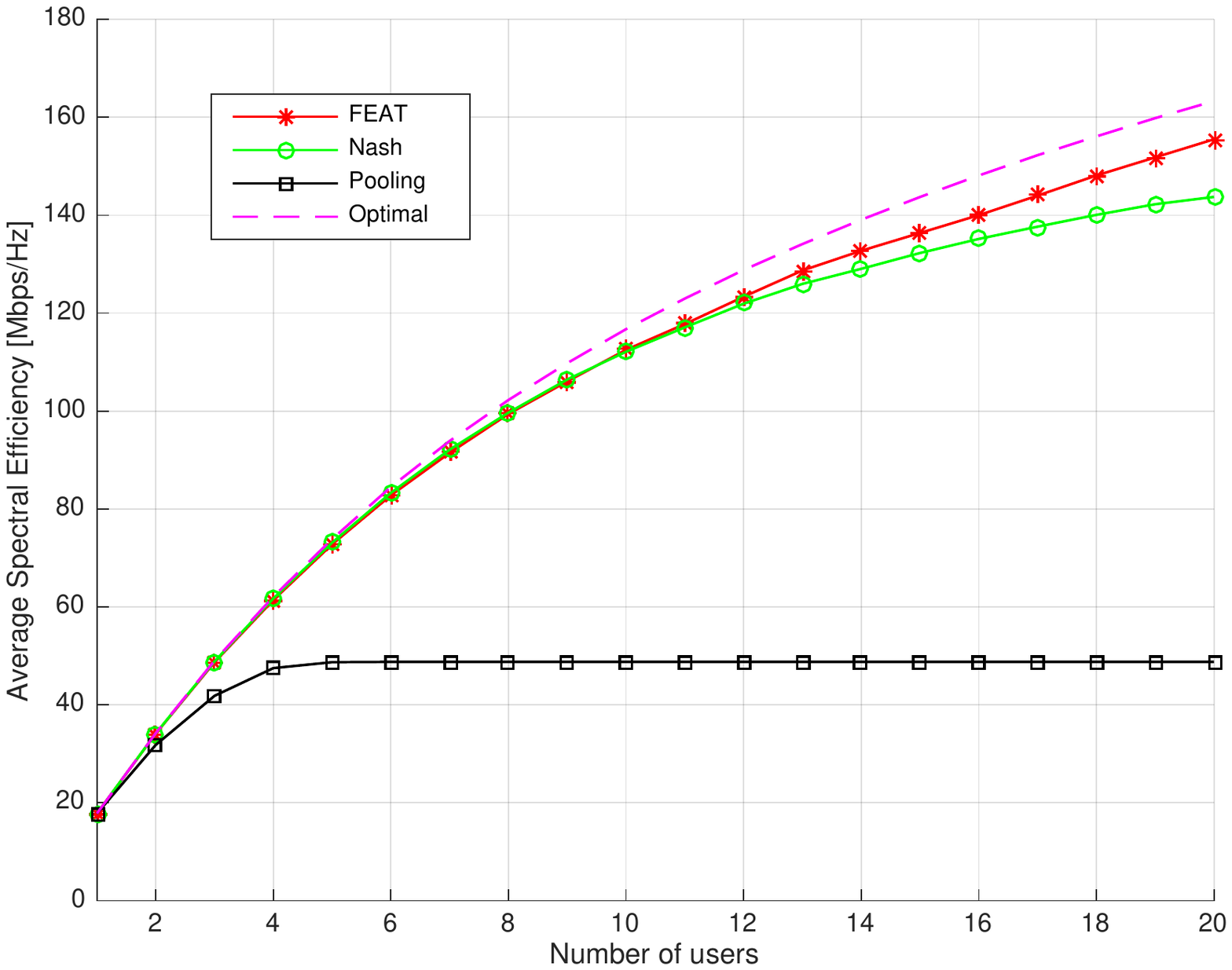}
\vspace{-3cm}
\caption{The average utility as function of $N$ for SNR $=10$ dB and $K=40$.}
\label{fig:SW_N_algo_K40_10dB}
\end{subfigure}
\end{figure}

\begin{figure}[t]
\centering
\vspace*{-4.5cm}
\hspace*{-0cm}
\includegraphics[height = 14cm,width=11cm]{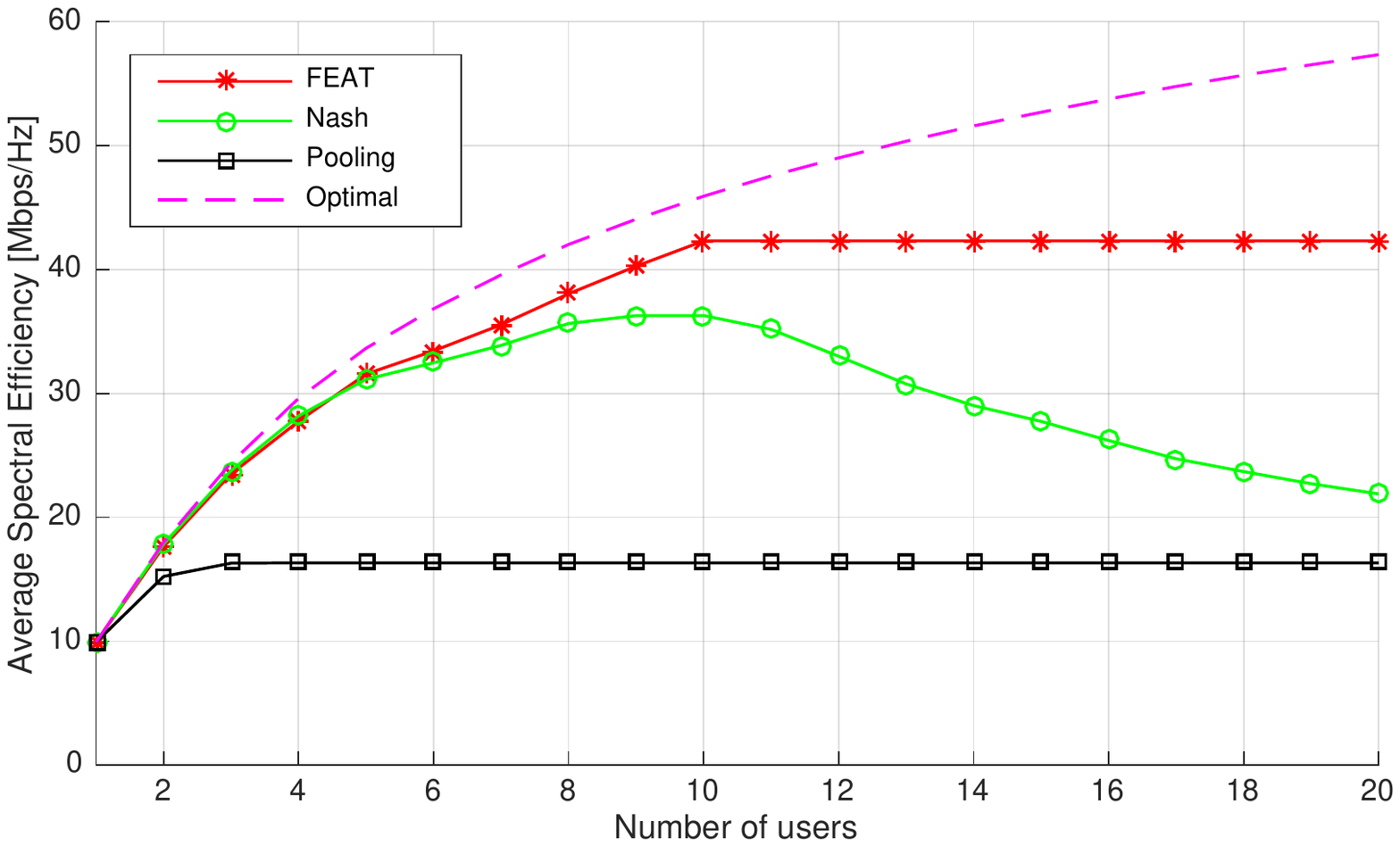}
\vspace{-4.5cm}
\caption{The average utility as function of $N$ for SNR $=10$ dB and $K=10$.}
\label{fig:SW_N_algo_K10}
\end{figure}

Fig. \ref{fig:SW_K_algo_-10dB} and Fig. \ref{fig:SW_K_algo_N20_10dB} depict the average spectral efficiency as function of the number of carriers $K$ for $N=20$ and SNR $=-10$ dB, respectively SNR $=10$ dB. At high SNR regime, users have less opportunities to transmit. So even though FEAT tries to assign at least one good channel to every user, some users may end up with average channels. This is responsible for the fact that the slope of social welfare tends to saturate much faster with decrease of $K$ in Fig. \ref{fig:SW_N_K}. In general, FEAT is better than Pooling and worse than Nash only if SNR is low (see Fig. \ref{fig:SW_K_algo_-10dB}). This is understandable, as in that case there are very good opportunities (in terms of spectral efficiency) to transmit. In such a situation, strictly avoiding interference (as we do in FEAT) is no more the most profitable thing to do. However, we emphasize that the range of SNR (see Fig. \ref{fig:SW_SNR_algo}) which gives upper hand to Nash is much smaller than that for FEAT. In general, notice that the average spectral efficiency at the NE is always surpassed by the optimal strategy, as already proved in Eq. (\ref{eq:util_ub}).

In Fig. \ref{fig:SW_N_algo_K40_10dB}, we see that for low $N$, FEAT is very close to the optimal solution for a low number of users. As the number of users $N$ increases, the opportunities to transmit becomes scarce giving rise to an interference-limited system. In that situation, FEAT outperforms Nash again, but the gap with the optimal strategy becomes high. This behavior confirms our theoretical claims in Section \ref{sec:opt}, where we showed that the utility in (\ref{eq:util_se}) is always bounded above by the fundamental limit on the individual information transmission rate.

Let us now look at the situation where the number of carriers $K$ decreases. In Fig. \ref{fig:SW_N_algo_K10}, we plot the average spectral efficiency as function of the number of users $N$ for SNR $=10$ dB and $K=10$. We notice that, as long as $K>N$, social welfare keeps increasing with $N$. Then, the social welfare does not grow above some threshold, but remains constant, contrary to Nash where spectral efficiency decreases as $N$ increases. Pooling gives all resources to first users in the queue, whereas in FEAT we try to share resources among all users in the system. In Pooling, when the number of users increases, all resources have been taken by first users in the queue and there is no addition in social welfare. In FEAT, even when we increase $N$, the social welfare keeps increasing because we give resources to users who will use them best and thus increase social welfare. 
As we can see from Fig. \ref{fig:SW_K_SNR_N20}, there are slight changes of slope when $K$ is a multiple of $N$. What we observe here is that, as long as $N \leq K$, social welfare increases at higher speed. This can be explained by the fact that FEAT starts by trying to allocate one good channel to each user. The channels it allocates to users afterwards are substantially worse, resulting in slower increase in social welfare. This can be clearly observed when $K=N=20$.

\begin{figure}[t]
\centering
\vspace*{-3cm}
\hspace*{-0cm}
\includegraphics[height = 12cm,width=10cm]{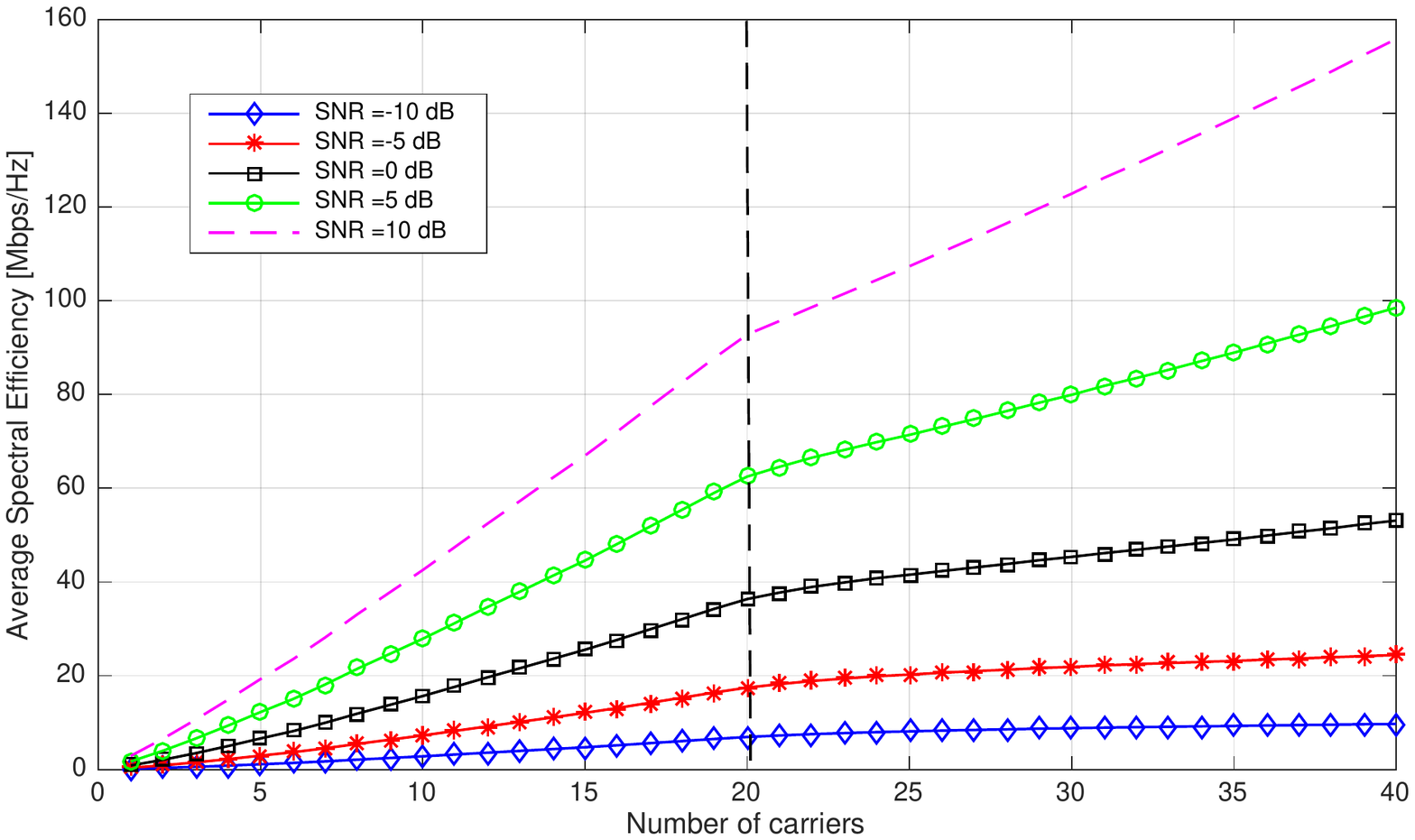}
\vspace{-3.5cm}
\caption{The average utility as function of $K$ for different SNR for $N=20$.}
\label{fig:SW_K_SNR_N20}
\end{figure}

\section*{Energy efficiency}
To study how energy efficient FEAT is, let us consider the following energy efficiency function that has bits per joule as units, and which perfectly captures the trade-off between throughput and battery life \cite{meshkati-jsac-2006} :
\beq\label{eq:ee}
EE_n(\mathbf{p_1},\ldots,\mathbf{p_{n}})=\frac{\displaystyle R_n \cdot \sum_{k=1}^K f(\gamma_{n}^k)}{\displaystyle\sum_{k=1}^K p_{n}^k},
\eeq
where $\mathbf{p_{n}}$ is the power control vector of user $n$ over all his carriers, \emph{i.e.}, $\mathbf{p_{n}}=(p_n^1,\ldots,p_n^K)$, $R_n$ stands for the transmission rate of user $n$ and $f(\cdot)$ is an S-shaped function, which measures the packet success rate. In Fig. \ref{fig:EE_N_algo_K10}, we plot the energy efficiency as function of $N$ for $K=10$ and a rate $R_n=1$ Mbps for each user $n$. As expected, it is shown that for low $N$, there will be little interference, so energy is used quite efficiently. This explains why Nash is close to FEAT in this region. Bigger $N$ gives rise to high interference, resulting in FEAT performing better in comparison to Nash. As the number of carriers $K$ increases (see Fig. \ref{fig:EE_N_algo_K40}), interference is no more an issue and Nash performs better, but remains less energy efficient than FEAT, especially for high number of users. Notice also that Pooling scheme is not energy efficient in general.
The main rationale behind these results is the way how FEAT deals with interference mitigation. As mentioned earlier, FEAT is more (spectral) efficient in situations where interference is an issue, as it turns off links which do not contribute enough information rate to outweigh the interference degradation caused by them to the rest of the system. We will see next that FEAT further ensures fairness among users by guaranteeing that every user in the system is assigned at least one good channel to every user in the system, as long as there are enough channels for every user.

\begin{figure}[t]
\centering
\vspace*{-3cm}
\hspace*{-0cm}
\includegraphics[height = 10.5cm,width=11cm]{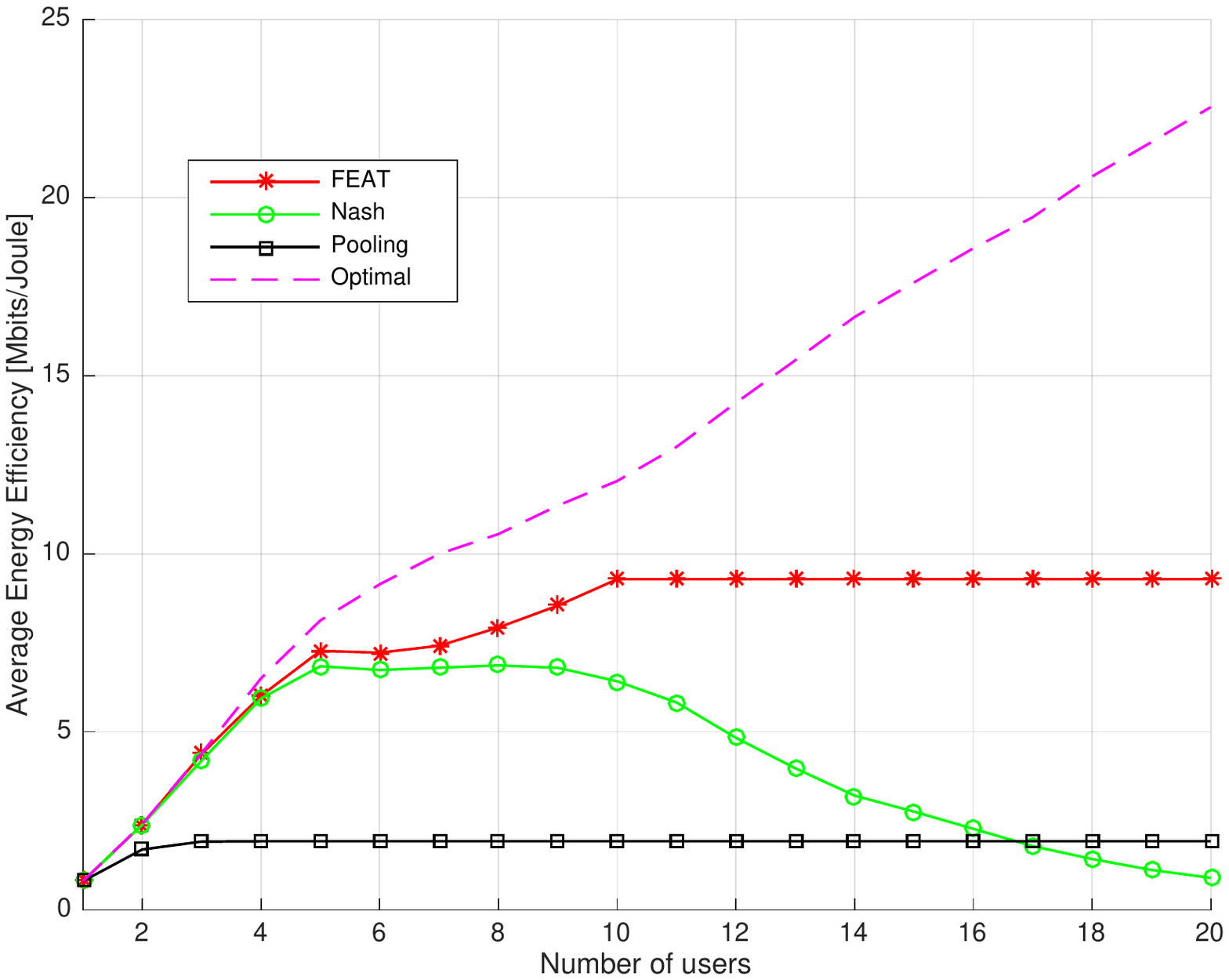}
\vspace{-3cm}
\caption{The average energy efficiency as function of $N$ for $K=10$.}
\label{fig:EE_N_algo_K10}
\end{figure}

\begin{figure}[t]
\centering
\vspace*{-3cm}
\hspace*{-0cm}
\includegraphics[height = 10.5cm,width=11cm]{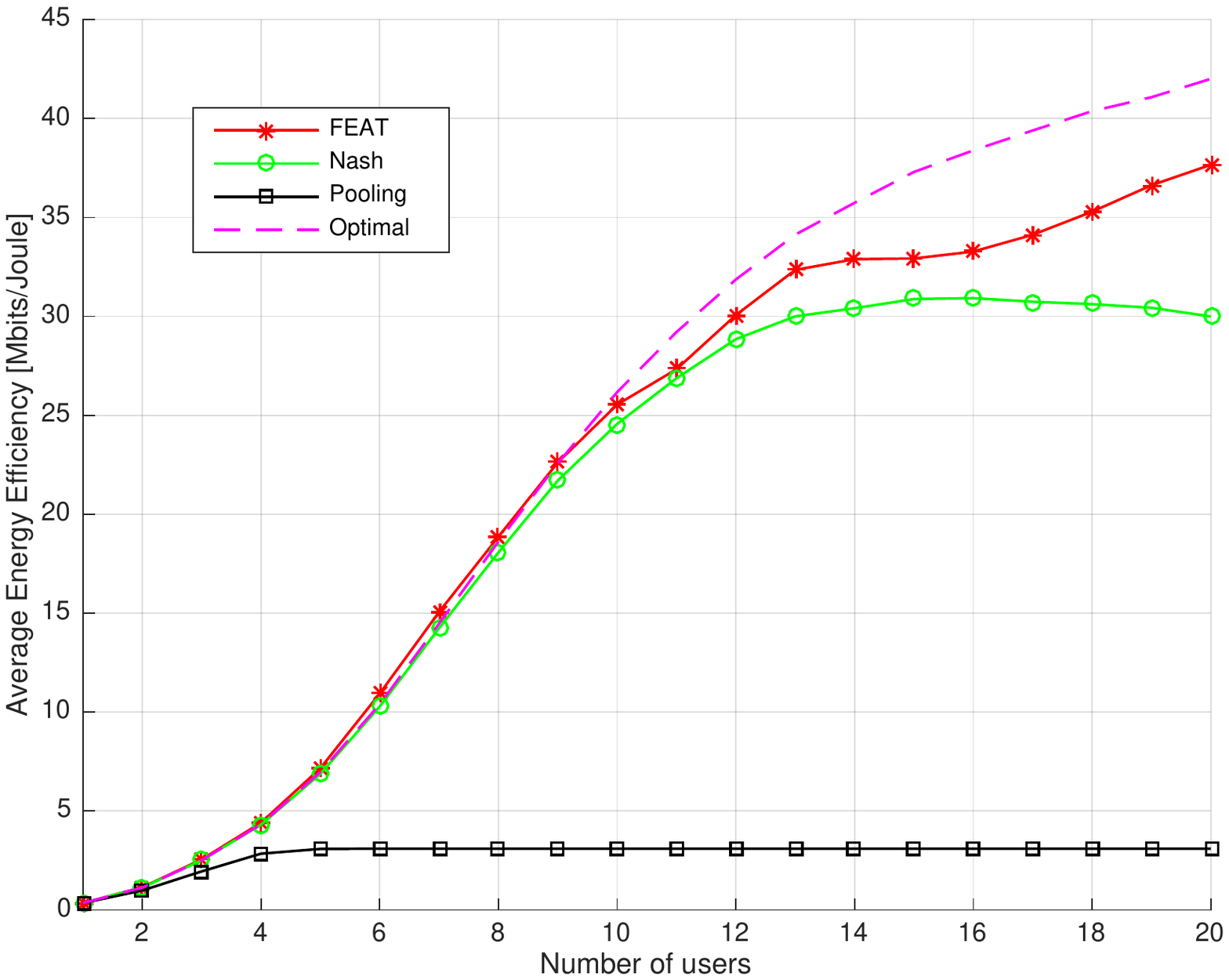}
\vspace{-3cm}
\caption{The average energy efficiency as function of $N$ for $K=40$.}
\label{fig:EE_N_algo_K40}
\end{figure}

\section*{Fairness}
So far, we have addressed the {\it system} performances of FEAT. To go further with the analysis and to show how fair FEAT is, we resort to {\it user} performance. Specifically, we plot in Fig. \ref{fig:fairness_N_algo_K40} the ratio between individual utility of the worst user and the best user in each of the schemes for $K=40$. As intuition would suggest, all curves decrease as $N$ increases, as there are less opportunities for users to pick their best channel. Clearly, FEAT is the best, whereas Nash and the optimal strategies become not fair for increasing number of users. 

For a low number of users, the ratio between individual utility of the worst user and the best user increases as the number of carriers increases (see Fig. \ref{fig:fairness_K_algo_N10_K40}). Here we remark that, when $K<N$, some users are not served by FEAT (although the algorithm works well for $K <N$), but once we reach the point where $K=N$, FEAT immediately starts to outperform all other schemes. We further see in Fig. \ref{fig:fairness_K_algo_N10_K40} that there are peaks anytime $K$ is divisible by $N$ because in that case each user will typically be assigned to the same number of carriers.

\begin{figure}[t]
\centering
\vspace*{-4.5cm}
\hspace*{-0cm}
\includegraphics[height = 14cm,width=11cm]{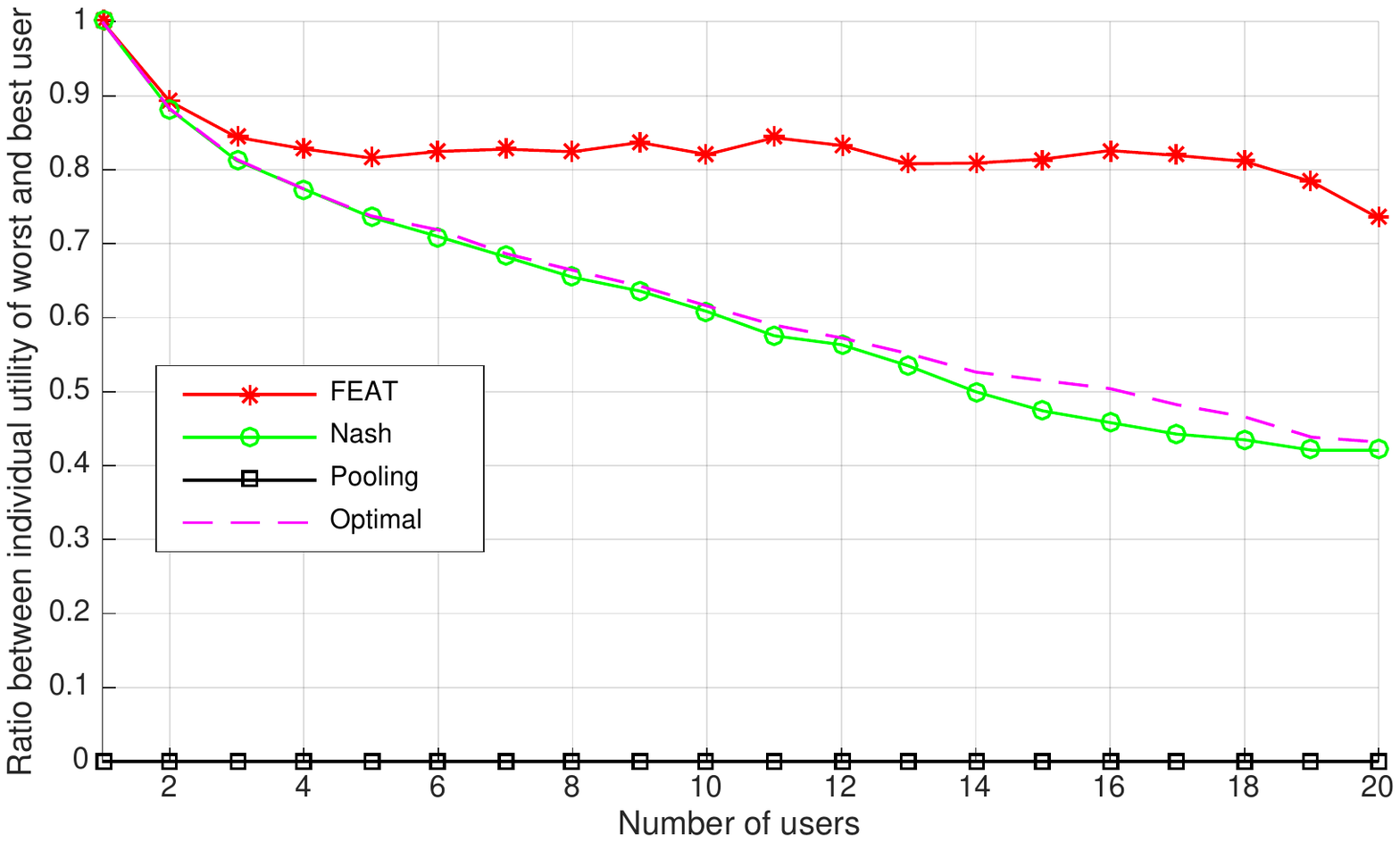}
\vspace{-4.5cm}
\caption{The ratio between individual utility of worst user and best user as function of $N$ for $K=40$.}
\label{fig:fairness_N_algo_K40}
\end{figure}

\begin{figure}
\centering
\vspace*{-3cm}
\hspace*{-0cm}
\includegraphics[height = 11cm,width=11cm]{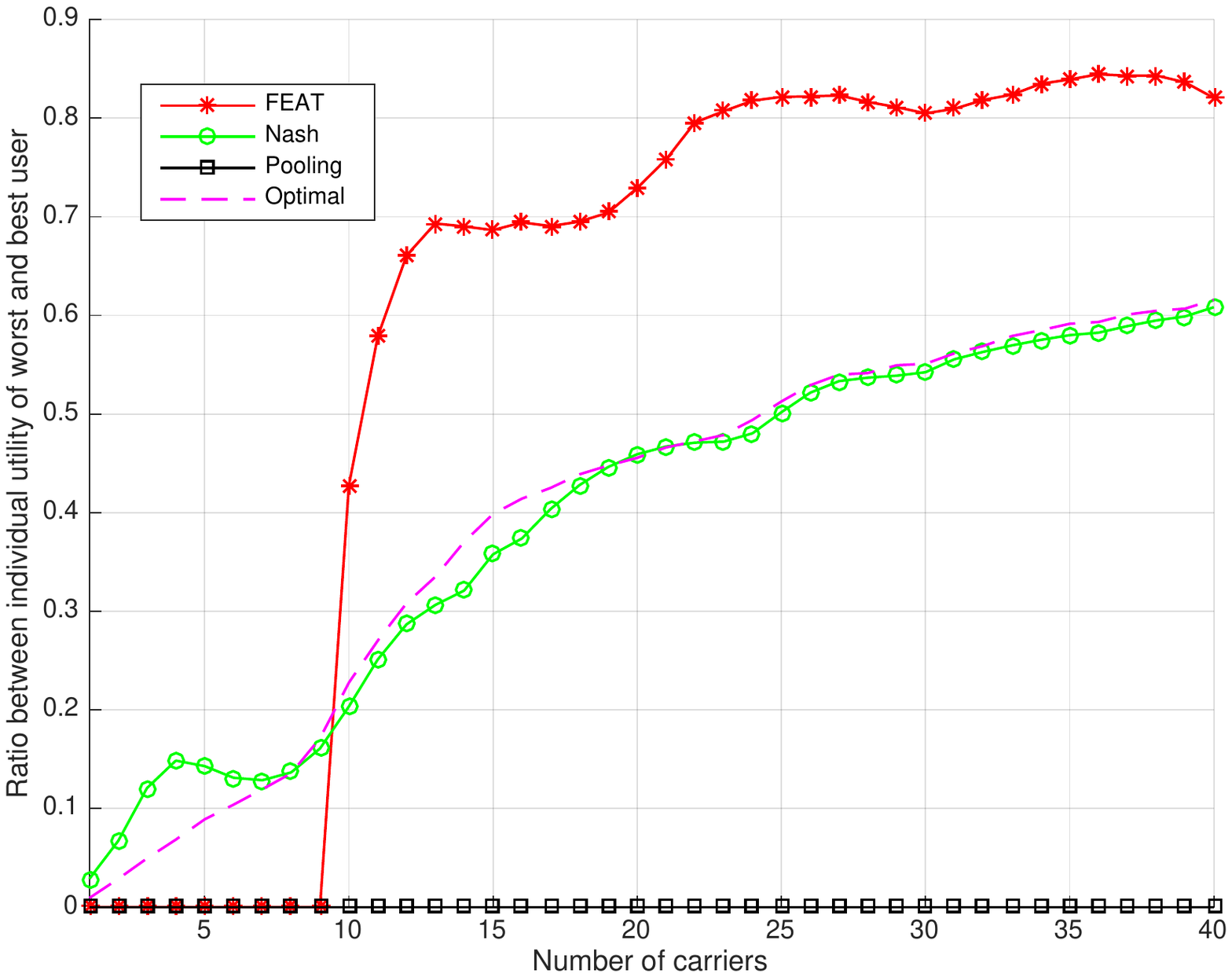}
\vspace{-3cm}
\caption{The ratio between individual utility of worst user and best user as function of $K$ for $N=10$.}
\label{fig:fairness_K_algo_N10_K40}
\end{figure}

\begin{figure}[t]
\centering
\vspace*{-3cm}
\hspace*{-0cm}
\includegraphics[height = 11cm,width=11cm]{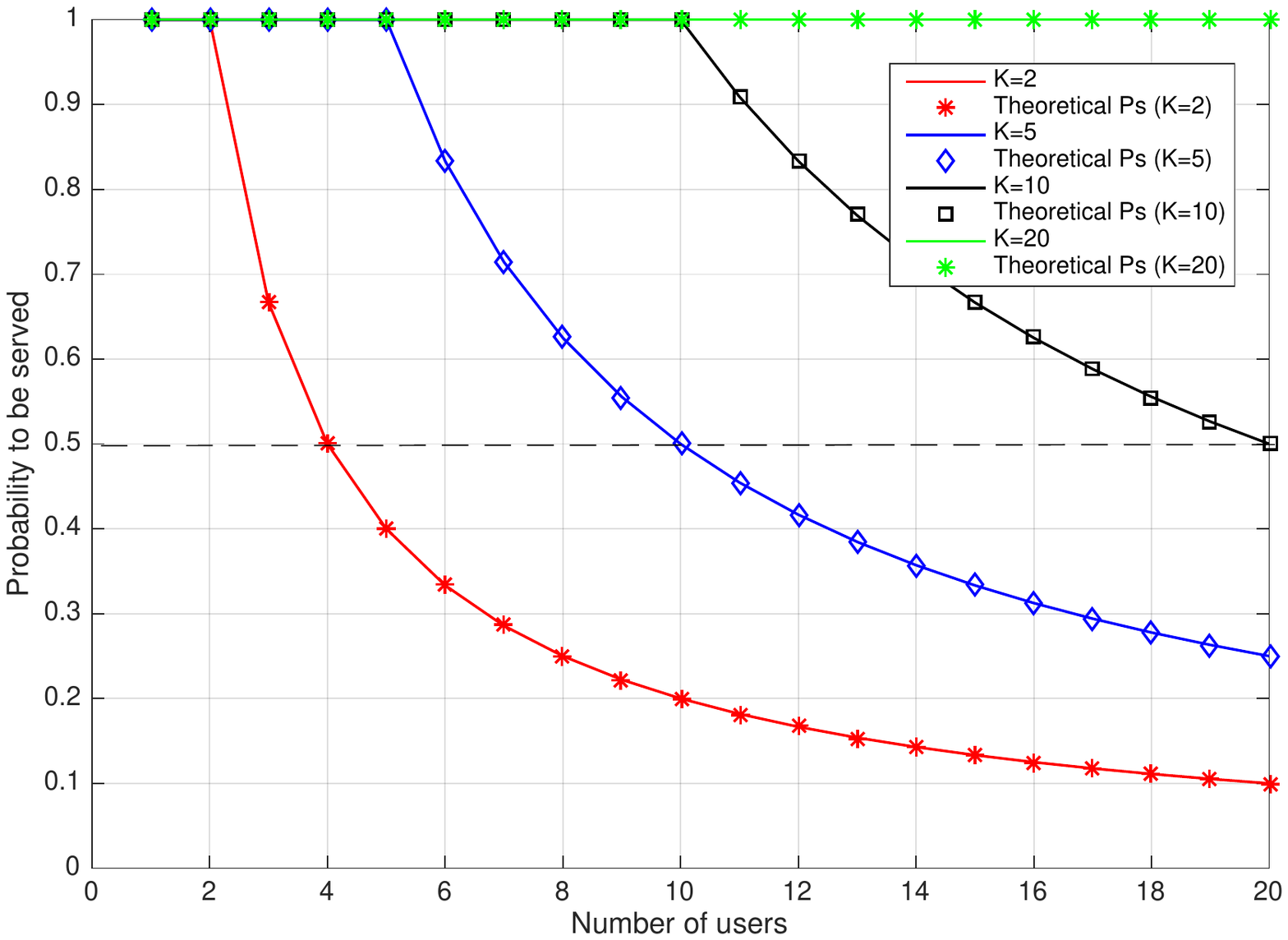}
\vspace{-3cm}
\caption{The probability to be served for FEAT as function of $N$ for different $K$.}
\label{fig:Ps_N_K_th}
 \end{figure}

Fig. \ref{fig:Ps_N_K_th} captures the probability to be served as function of $N$ for different $K$: it is equal to 0.5 (\emph{i.e.}, half of users are allocated at least one carrier to transmit on) when $\frac{K}{N}=0.5$. In fact, as long as $\frac{K}{N}\geq1$ the probability to be served equals $1$, otherwise it is equal to $\frac{K}{N}$. As expected, the probability to be served increases with $K$ and decreases with $N$. We emphasize that there is no possibility of lack of transmission whenever\footnote{Notice that this is not a restriction of the proposed scheme. In fact, FEAT always produces some output even for $K<N$.} $K>N$. Then, each user transmits on at least one carrier and the sets of carriers they transmit on are orthogonal (as we can see from Fig. \ref{fig:Pnc_N_algo_K10_N20}).

\begin{figure}[t]
\centering
\vspace*{-3cm}
\hspace*{-0cm}
\includegraphics[height = 11cm,width=11cm]{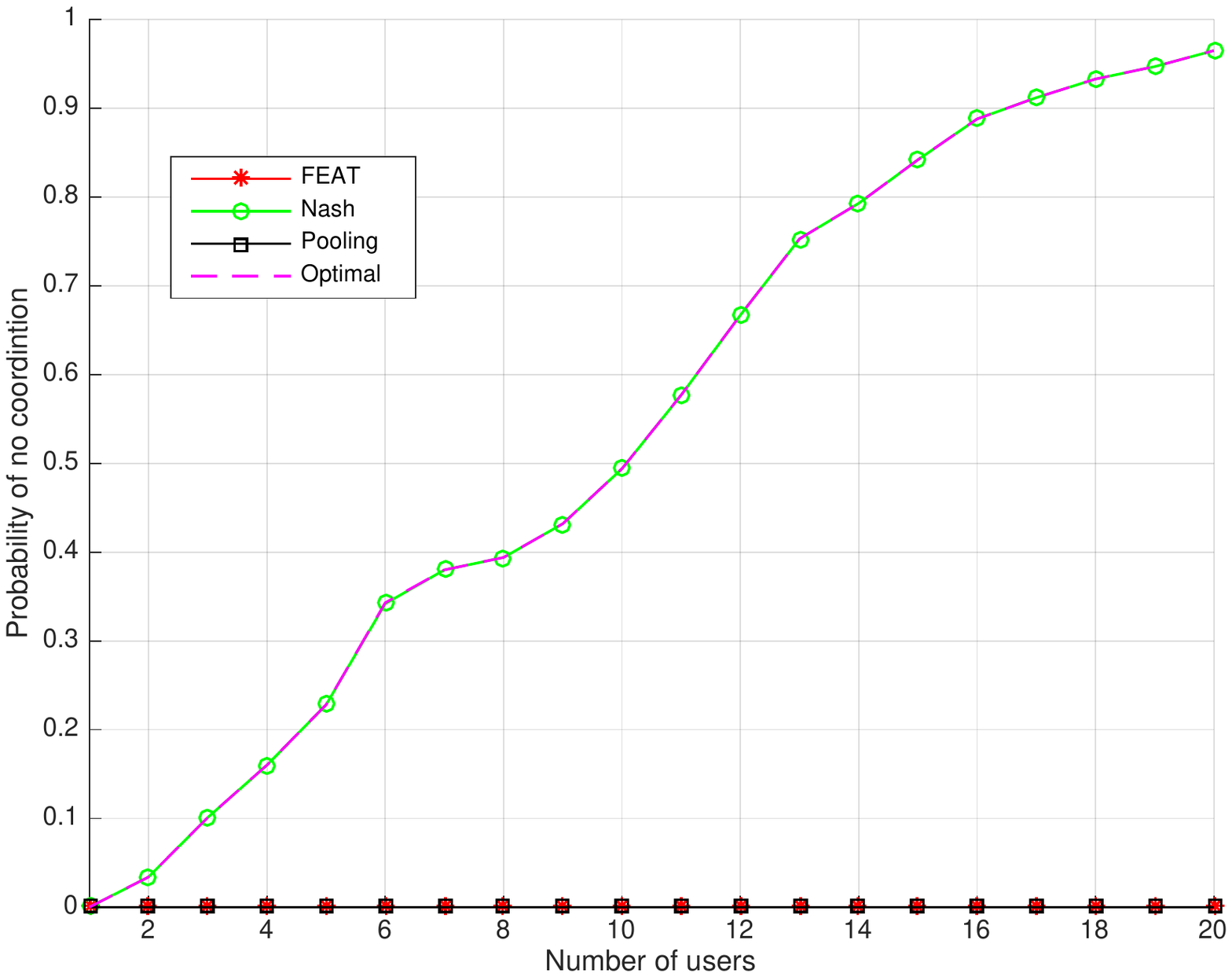}
\vspace*{-3cm}
\caption{The probability of no coordination as function of $N$ for $K=10$.}
\label{fig:Pnc_N_algo_K10_N20}
 \end{figure}

\section*{Robustness}

In order to get good intuition on the robustness of FEAT against individual incentive to deviate, we resort to compute the ratio between the utility of a random user if he sticks to the output of FEAT and his utility when he waterfills on all the channels. This will tell us how much incentive may a random user have to deflect from the output of FEAT, which in fact corresponds to the measure of the "distance" between the players' joint strategy profile using FEAT and a Nash equilibrium. Accordingly, a ratio close to $1$ means that users cannot deviate from FEAT and gain, whereas getting further away from $1$ corresponds to a situation with a bigger gain from individual deviation. So the closer the ratio gets to $1$ the better it is. This is done in Fig. \ref{fig:gap_K_N} for different $K$, because Prop. \ref{prop:algo_properties} and corollaries after it are talking about dependence of this "distance" from the social optimum and from NE as a function of $K$. The results seem to confirm our expectations. What we observe for $K\leq N$ is that users have the same incentive to deviate and that the decrease of the ratio is rather fast. The explanation of this is simple: When $K<N$, FEAT only assigns first $K$ users to channels -- all the remaining ones remain unassigned, which means that with probability $\frac{N-K}{N}$ they can increase their utility from 0 to some positive value, while with remaining probability the maximum increase in their utility is limited. For $K\geq N$, we observe that the distance from Nash is almost constant and dependent (in a decreasing manner) on $N$. The stabilization of it is what we expected, as explained in Remark \ref{rem:last}. The dependence on $N$ can be explained by a similar dependence of $\alpha^*_1$ (appearing in the denominator of the bound given in Prop. 2) on the number of users. This is rather expected as the degree of freedom grows with $N$ and/or $K$, resulting in increasing the users' opportunity to deviate by taking advantage of the multiuser diversity gain. To sum it up, what we remark here is that fairness does not imply \emph{a fortiori} little incentive to deviate.

\begin{figure}[t]
\centering
\vspace*{-3cm}
\hspace*{-0cm}
\includegraphics[height = 11cm,width=11cm]{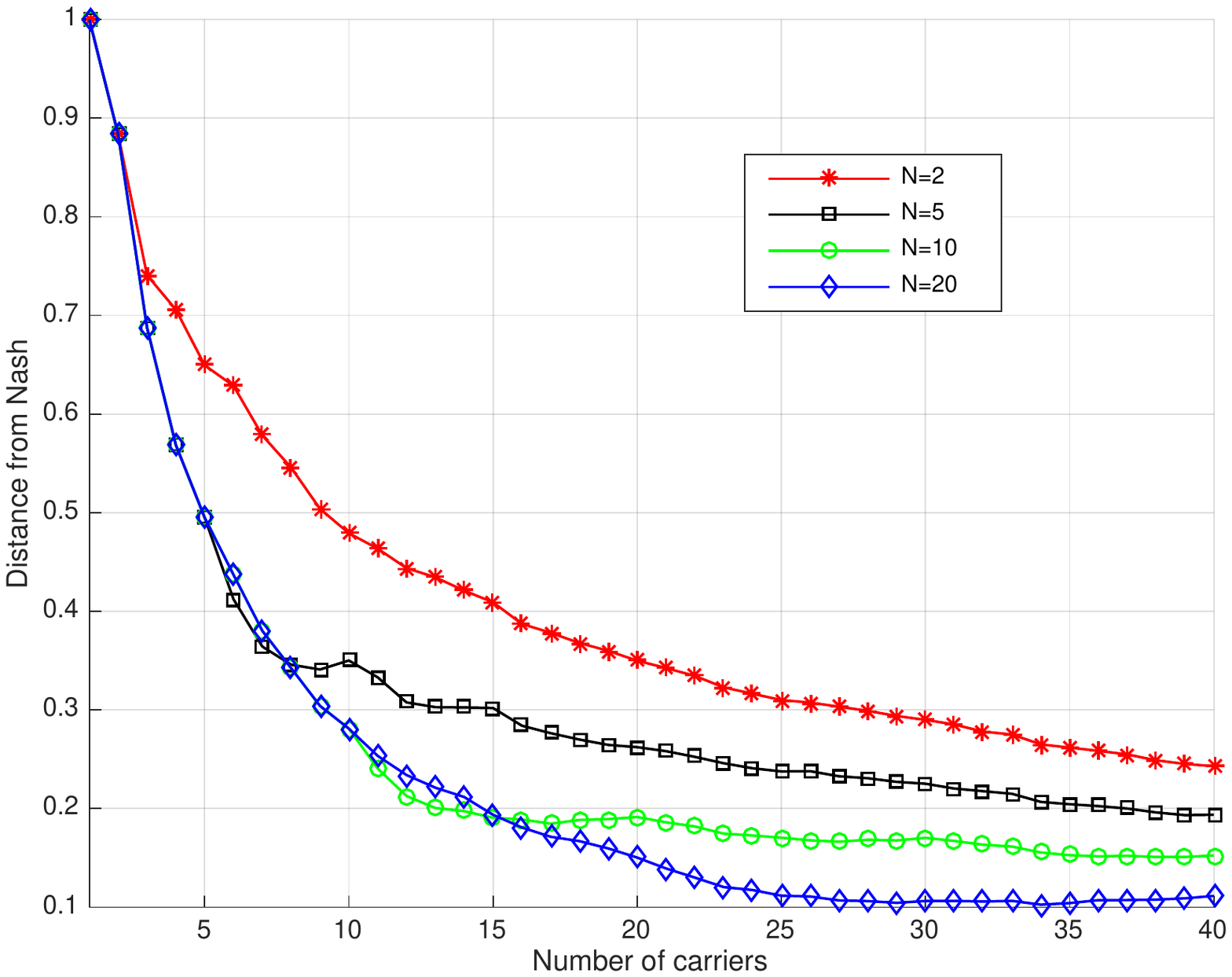}
\vspace{-3cm}
\caption{The ratio between individual utility without and with deflection from FEAT as function of $K$ for different $N$. }
\label{fig:gap_K_N}
\end{figure}

\section{Conclusion}\label{sec:conc}
We have addressed the information rate maximization problem of Gaussian sources by proposing a joint channel selection and power allocation scheme based on the water-filling algorithm. Contrary to {\it{classical}} water filling solutions in the literature, the way we have done this is by eliminating/adding channels on which we could have transmitted in the case without FEAT. In this aspect, FEAT is similar to mercury water-filling, but the way it is done is not by changing the water-filling level (as it is done with mercury water-filling), but by limiting the available carriers for each user. While maximizing the information rate, FEAT also guarantees fairness among users by trying to assign at least one good channel to every user in the system, as long as there are enough channels for every user. It has been shown that FEAT outperforms all other algorithms, especially in interference-limited systems. With FEAT, we can notably ensure a near-optimal solution with low computational complexity. In particular, FEAT does not require to implement a SIC scheme at the receiver side, as the optimal scheme does.


\section{appendix}

\label{proof:opt}
Given a power allocation strategy profile $\textbf{p}\in \mathbb{R}^{N\times K}$, for all achievable rates $\left\{R_1, R_2, \ldots, R_N \right\}$, the following inequality holds for the fundamental limit on the sum information transmission rate \cite{cover}:
\begin{eqnarray}
 \sum_{n=1}^{N} R_n & \leq & \sum_{k=1}^{K} \log_2\left(1+\frac{\displaystyle\sum_{n=1}^{N} g_{n}^k p_{n}^k}{\sigma^2}
\right)
\label{eq:util_rn}
\end{eqnarray}

On the other hand, we have that\\
\begin{eqnarray}
\sum_{n=1}^{N}\sum_{k=1}^{K} \log_2 \left(1+\frac{g_{n}^k p_{n}^k}{\sigma^2+ \displaystyle\sum_{\substack{m=1 \\m<n}}^N g_{m}^k p_{m}^k}\right)
&=&\sum_{k=1}^{K} \log_2 \left(1+\frac{g_{1}^k p_{1}^k}{\sigma^2}\right)+ 
\sum_{k=1}^{K}\log_2 \left(1+\frac{g_{2}^k p_{2}^k}{\sigma^2+ \displaystyle g_{1}^k p_{1}^k}\right)+ \cdots \nonumber\\
&&+\sum_{k=1}^{K}\log_2 \left(1+\frac{g_{N}^k p_{N}^k}{\sigma^2+ \displaystyle \sum_{m=1}^{N-1} g_{m}^k p_{m}^k}\right),\nonumber
\end{eqnarray}
which gives after some simplifications \\
\begin{eqnarray}
\sum_{n=1}^{N}\sum_{k=1}^{K} \log_2 \left(1+\frac{g_{n}^k p_{n}^k}{\sigma^2+ \displaystyle\sum_{\substack{m=1 \\m<n}}^N g_{m}^k p_{m}^k}\right)
&=& 
\sum_{k=1}^{K} \log_2\left(1+\frac{\displaystyle\sum_{n=1}^{N} g_{n}^k p_{n}^k}{\sigma^2}
\right).
\label{eq:util_unk}
\end{eqnarray}\\

Now, combining Eq. (\ref{eq:util_rn}) and Eq. (\ref{eq:util_unk}), we obtain the following inequality on the fundamental limit on the sum information transmission rate, namely
\begin{eqnarray}
 \sum_{n=1}^{N} R_n & \leq & \sum_{n=1}^{N}\sum_{k=1}^{K} \log_2 \left(1+\frac{g_{n}^k p_{n}^k}{\sigma^2+ \displaystyle\sum_{\substack{m=1 \\m<n}}^N g_{m}^k p_{m}^k}\right).
 \label{eq:util_rnn}
\end{eqnarray}

Then, the following inequality holds on the fundamental limit on the individual information transmission rate, namely

\beq
u_n(\textbf{p})=\sum_{k=1}^{K} \log_2 \left(1+\frac{g_{n}^k p_{n}^k}{\sigma^2+ \displaystyle\sum_{\substack{m=1 \\
m\neq n}}^{N} g_{m}^k p_{m}^k}\right)
\leq \sum_{k=1}^{K} \log_2 \left(1+\frac{g_{n}^k p_{n}^k}{\sigma^2+ \displaystyle\sum_{\substack{m=1 \\m<n}}^N g_{m}^k p_{m}^k}\right),
\label{eq:proof_util_1}
\eeq\\\\
because for all $(n,k) \in \left\{1, 2, \ldots, N \right\} \times \left\{1, 2, \ldots, K \right\}$\\

\beq
\log_2 \left(1+\frac{g_{n}^k p_{n}^k}{\sigma^2+ \displaystyle\sum_{\substack{m=1 \\
m\neq n}}^{N} g_{m}^k p_{m}^k}\right)
\leq 
\log_2 \left(1+\frac{g_{n}^k p_{n}^k}{\sigma^2+ \displaystyle\sum_{\substack{m=1 \\m<n}}^N g_{m}^k p_{m}^k}\right).
\label{eq:proof_util_2}
\eeq\\

\bibliographystyle{IEEEtran}
\bibliography{biblio.bib}

\end{document}